\documentclass{article}
\usepackage[utf8]{inputenc}

\usepackage{fullpage}

\usepackage{amsthm,amsmath,amssymb,booktabs,xcolor,graphicx}
\usepackage{thm-restate}
\usepackage{bbm}
\usepackage{listings}
\usepackage{xcolor}
\usepackage{appendix}
\usepackage{enumerate}
\usepackage{mathbbol}
\usepackage[numbers]{natbib}
\usepackage{comment}
\usepackage[shortlabels]{enumitem}

\usepackage[margin=1in]{geometry}
\usepackage[utf8]{inputenc}
\usepackage[T1]{fontenc}

\usepackage[textsize=small,backgroundcolor=orange!20]{todonotes}

\definecolor{darkgreen}{rgb}{0,0.5,0}
\usepackage{hyperref}
\hypersetup{
    unicode=false,          
    colorlinks=true,        
    linkcolor=red,          
    citecolor=darkgreen,        
    filecolor=magenta,      
    urlcolor=cyan           
}

\usepackage{url}
\usepackage{etoolbox}
\usepackage{appendix}
\usepackage{xspace}
 \usepackage{algorithm}
\usepackage[noend]{algpseudocode}
\usepackage[labelfont=bf]{caption}
\usepackage[capitalize, nameinlink, noabbrev]{cleveref}
\crefname{equation}{}{} 
\AtBeginEnvironment{appendices}{\crefalias{section}{appendix}} 

\usepackage[color,final]{showkeys} 

\colorlet{refkey}{orange!20}
\colorlet{labelkey}{blue!30}

\numberwithin{equation}{section}
\newtheorem{theorem}{Theorem}[section]

\newtheorem{lemma}[theorem]{Lemma}
\newtheorem{claim}[theorem]{Claim}

\crefname{claim}{Claim}{Claims}

\newtheorem{corollary}[theorem]{Corollary}

\newtheorem*{question*}{Question}

\theoremstyle{definition}
\newtheorem{definition}[theorem]{Definition}

\newtheorem*{definition*}{Definition}

\theoremstyle{remark}

\newcommand{\set}[1]{\left\{ #1 \right\}}

\newcommand{\poly}{\mathrm{poly}}

\newcommand{\fC}{\mathcal{C}}

\newcommand{\eps}{\varepsilon}

\newcommand{\E}{\mathbb{E}}

\newcommand{\diam}{\textrm{diam}}

\newcommand{\tO}{\widetilde{O}}

\newcommand{\fA}{\mathcal{A}}

\newcommand{\Labels}{\Sigma}
\newcommand{\polylog}{\textrm{polylog}}

\newcommand{\congest}{$\mathsf{CONGEST}\,$}
\newcommand{\local}{$\mathsf{LOCAL}\,$}
\newcommand{\pram}{$\mathsf{PRAM}\,$}

\newcommand{\dead}{\textrm{dead}}
\newcommand{\alive}{\textrm{alive}}
\newcommand{\act}{\textrm{active}}
\newcommand{\cut}{\textrm{cut}}
\newcommand{\expand}{\textrm{expand}}

\newcommand{\utility}{\mathbf{u}}
\newcommand{\cost}{\mathbf{c}}

\newcommand{\del}{\mathrm{del}}
\newcommand{\wait}{\mathrm{wait}}
\newcommand{\frontier}{\mathrm{frontier}}

\newcommand{\separation}{\text{s}}

\newcommand{\hfC}{\hat{\fC}}

\title{Improved Distributed Network Decomposition, \\ Hitting Sets, and Spanners, via Derandomization}
\author{
  Mohsen Ghaffari \\
  \small{MIT}\\
  \small{ghaffari@mit.edu}
  \and
  Christoph Grunau \\
  \small{ETH Zurich}\\
  \small{cgrunau@inf.ethz.ch}
  \and
  Bernhard Haeupler \\
  \small{ETH Zurich and CMU}\\
  \small{bernhard.haeupler@inf.ethz.ch}
  \and
  Saeed Ilchi  \\
  \small{ETH Zurich} \\
  \small{saeed.ilchi@inf.ethz.ch}
  \and
  Václav Rozhoň \\
\small{ETH Zurich} \\
\small{rozhonv@inf.ethz.ch}
}
\date{}
\begin{document}
\maketitle
\begin{abstract}
This paper presents significantly improved deterministic algorithms for some of the key problems in the area of distributed graph algorithms, including network decomposition, hitting sets, and spanners. As the main ingredient in these results, we develop novel randomized distributed algorithms that we can analyze using only pairwise independence, and we can thus derandomize efficiently. As our most prominent end-result, we obtain a deterministic construction for $O(\log n)$-color $O(\log n \cdot \log\log\log n)$-strong diameter network decomposition in $\tO(\log^3 n)$ rounds. This is the first construction that achieves almost $\log n$ in both parameters, and it improves on a recent line of exciting progress on deterministic distributed network decompositions [Rozho\v{n}, Ghaffari STOC'20; Ghaffari, Grunau, Rozho\v{n} SODA'21; Chang, Ghaffari PODC'21; Elkin, Haeupler, Rozho\v{n}, Grunau FOCS'22].
\end{abstract}

{
\setcounter{page}{0}
\thispagestyle{empty}
\newpage}

{ 
\bigskip
\hypersetup{linkcolor=blue}
\tableofcontents
\setcounter{page}{0}
\thispagestyle{empty}
}

\newpage
\section{Introduction}
This paper is centered on the area of \emph{distributed graph algorithms} and provides new methods and tools for developing improved \emph{deterministic} distributed algorithms. 

It has been a central, well-known, and well-studied theme in this area that, for many of the graph problems of interest, known randomized algorithms outperform their deterministic counterparts. Concretely, the randomized variants have been much faster and/or achieved better output properties, e.g., approximation factors. As a prominent example, for several of the key problems of interest---including maximal independent set, maximal matching, $\Delta+1$ vertex coloring---we have known $O(\log n)$ round randomized algorithms since the 1986 work of Luby~\cite{luby86}. In contrast, developing even $\poly(\log n)$-time deterministic algorithms for many of these problems remained open for nearly four decades. See for instance the 2013 book of Barenboim and Elkin~\cite{barenboimelkin_book} which lists numerous such open questions. Only very recently, $\poly(\log n)$-time deterministic algorithms for these problems were developed~\cite{rozhonghaffari20, GGR20, chang2021strong, GhaffariK21, elkin2022deterministic}. However, currently, these deterministic algorithms are still quite far from their randomized counterparts.

In this paper, we focus on two of the most central tools in developing deterministic algorithms for local graph problems, namely \emph{network decompositions} and \emph{hitting sets}, and we present significantly improved deterministic distributed constructions of these tools. From a technical perspective, our novelty is in developing new randomized algorithms for these tools in such a way that we can analyze the algorithm by assuming only pairwise independence in the randomness it uses. We then describe how one can leverage this to derandomize the algorithms, i.e., to transform the randomized algorithm into an efficient deterministic algorithm. We next review the model and then state our contributions in the context of the recent progress. 

\paragraph{Model.} We work with the standard distributed message-passing model for graph algorithms~\cite{peleg00}. The network is abstracted as an $n$-node graph $G=(V, E)$ where each node $v\in V$ corresponds to one processor in the network. Communications take place in synchronous rounds. Per round, each processor/node can send an $O(\log n)$-bit message to each of its neighbors in $G$. This model is called \congest. The relaxed variant of the model where we allow unbounded message sizes is called \local. At the end of the round, each processor/node performs some computations on the data that it holds, before we proceed to the next communication round. 

A graph problem in this model is captured as follows: Initially, the network topology is not known to the nodes of the graph, except that each node $v\in V$ knows its own unique $O(\log n)$-bit identifier and perhaps some of the global parameters of the network, e.g., the number $n$ of nodes in the network or a suitably tight upper bound on it. At the end of the computation, each node should know its own part of the output, e.g., in the graph coloring problem, each node should know its own color. When we discuss a particular graph problem, we will specify what part of the output should be known by each node.

\subsection{Network Decomposition}
Perhaps the most central object in the study of deterministic distributed algorithms for local graph problems has been the concept of \emph{network decomposition}, which was introduced by Awerbuch, Luby, Goldberg, and Plotkin~\cite{awerbuch89}. We next define this concept and explain its usefulness. Then, we discuss its existence and randomized distributed constructions. Afterward, we review the deterministic distributed constructions, especially the recent breakthroughs, and state our contributions.

Generally, the vertices of any $n$-node network can be colored using $O(\log n)$ colors such that in the subgraph induced by each color, each connected component has  diameter $O(\log n)$. We call this an $O(\log n)$-color $O(\log n)$-diameter network decomposition (or sometimes $O(\log n)$-color $O(\log n)$-strong-diameter network decomposition, to contrast it with a weaker variant which we discuss later). This decomposition enables us to think of the entire graph as a collection of $O(\log n)$ node-disjoint graphs, each of which has a small $O(\log n)$-diameter per component; the latter facilitates distributed coordination and computation in the component. 
As a prototypical example, given such a network decomposition, one easily gets an $O(\log^2 n)$-round deterministic algorithm for maximal independent set in the \local model: we process the color classes one by one, and per color, in each $O(\log n)$-diameter component, we add to the output a maximal independent set of the nodes of the component that do not have a neighbor in the independent sets computed in the previous colors. Each color is processed in $O(\log n)$ rounds, as that is the component diameter, and thus the overall process takes $O(\log^2 n)$ rounds. See \cite{rozhonghaffari20, ghaffari2017complexity, ghaffari2018derandomizing} for how network decomposition leads to a general derandomization method in the \local model, which transforms any $\poly(\log n)$-time randomized algorithm for any locally checkable problem~\cite{naor95} (roughly speaking, problems in which any proposed solution can be checked deterministically in $\poly(\log n)$-time, e.g., coloring, maximal independent set, maximal matching) into a $\poly(\log n)$-time deterministic algorithm.  

The existence of such a $O(\log n)$-color $O(\log n)$-diameter network decomposition follows by a simple ball-growing process~\cite{Awerbuch-Peleg1990}.
We build the colors one by one, and each time, we color at least half of the remaining nodes with the next color. For one color $i$, start from an arbitrary node and grow its ball hop by hop, so long as the size is increasing by at least a $2$ factor per hop. This stops in at most $O(\log n)$ hops. Once stopped, color the inside of the ball with the current color $i$, and remove the boundary nodes, deferring them to the next colors. If we continue doing this from nodes that remain in the graph, in the end, at least half of the nodes of the graph (which remained after colors $1$ to $i-1$) are colored in this color $i$, each carved ball has diameter $O(\log n)$, and different balls are non-adjacent as we remove their boundaries. 

 Linial and Saks~\cite{linial92} gave a randomized distributed algorithm that computes almost such a network decomposition in $O(\log^2 n)$ rounds of the \congest model. The only weakness was in the diameter guarantee: the vertices of each color are partitioned into non-adjacent clusters so that per cluster, every two vertices of this cluster have a distance of at most $O(\log n)$ in the original graph. This is what we call $O(\log n)$ weak-diameter. In contrast, if the distance was measured in the subgraph induced by the nodes of this color, it is called a \emph{strong-diameter}. A $O(\log^2 n)$-round \congest-model randomized algorithm for $O(\log n)$-color $O(\log n)$-strong-diameter network decomposition was provided much later, by Elkin and Neiman~\cite{elkin16_decomp}, building on a parallel algorithm of Miller, Peng, and Xu~\cite{miller2013parallel}.

 In contrast, even after significant recent breakthroughs, deterministic constructions for network decomposition are still far from achieving similar measures, and this suboptimality spreads to essentially all applications of network decomposition in deterministic algorithms. The original work of Awerbuch et al.~\cite{awerbuch89} gave a $T$-round deterministic \local algorithm for $c$-color and $d$-strong-diameter network decomposition where $c=d=T=2^{O(\sqrt{\log n\log\log n})}$. All these bounds were improved to $c=d=T=2^{O(\sqrt{\log n})}$ by Panconesi and Srinivasan~\cite{panconesi-srinivasan}. A transformation of Awerbuch et al.~\cite{awerbuch96} in the \local model can transform these into a $O(\log n)$-color $O(\log n)$-strong-diameter network decomposition, but the time complexity remains $2^{O(\sqrt{\log n})}$ and this remained the state of the art for over nearly three decades.
 
Rozho\v{n} and Ghaffari~\cite{rozhonghaffari20} gave the first $\poly(\log n)$ time deterministic network decomposition  with $\poly(\log n)$ parameters. Concretely, their algorithm computes a $O(\log n)$-color $O(\log^3 n)$-weak-diameter network decomposition in $O(\log^8 n)$ rounds of the \congest model. The construction was improved to a $O(\log n)$-color $O(\log^2 n)$-weak-diameter network decomposition in $O(\log^5 n)$ rounds of the \congest model, by Grunau, Ghaffari, and Rozho\v{n}~\cite{GGR20}. Both of these constructions were limited to only a weak-diameter guarantee. If one moves to the relaxed \local model with unbounded message sizes, then by combining these with a known transformation of Awerbuch et al.~\cite{awerbuch96}, one gets $O(\log n)$-color $O(\log n)$-strong-diameter network decompositions, in a time complexity that is slower by a few logarithmic factors. However, such a transformation was not known for the \congest model, until a recent work of Chang and Ghaffari~\cite{chang2021strong}. They gave a \congest-model reduction, which can transform the weak-diameter construction algorithm of Grunau et al.~\cite{GGR20} into a strong-diameter one, sacrificing some extra logarithmic factors. Concretely, they achieved a $O(\log n)$-color $O(\log^2 n)$-strong-diameter decomposition in $O(\log^{11} n)$ rounds. The time complexity of decomposition with these parameters was improved very recently by Elkin et al.~\cite{elkin2022deterministic}, obtaining a $O(\log n)$-color $O(\log^2 n)$-strong-diameter decomposition in $O(\log^5 n)$ rounds.

However, all these constructions are still far from building the arguably right object, i.e., an $O(\log n)$-color $O(\log n)$-strong-diameter decomposition, in the \congest model. This was true even if we significantly relax the time complexity, and as mentioned before, this sub-optimality spreads to all applications.

\paragraph{Our contribution.} In this paper, we present a novel deterministic construction of network decomposition which builds \emph{almost the right object}, achieving an $O(\log n)$-color $O(\log n \cdot \log\log\log n)$-strong-diameter decomposition, in $\poly(\log n)$ rounds. We note that all previous construction techniques seem to require diameter at least $\Omega(\log^2 n)$; see \cite{chang2021strong} for an informal discussion on this. Our algorithm breaks this barrier and reaches diameter $O(\log n \cdot \log\log\log n)$. The key novelty is in designing a new randomized algorithm that can be analyzed using only pairwise independence. We can thus derandomize this algorithm efficiently by using previously known network decompositions in a black-box manner, and in $\poly(\log n)$ time. 

Furthermore, if we want faster algorithms, by a black-box combination of our new construction with the technically-independent recent work of Faour, Ghaffari, Grunau, Kuhn, and Rozho\v{n}~\cite{Faour2022} on locally derandomizing pairwise-analyzed randomized algorithms (roughly speaking, their approach works by a specialized weighted defective coloring, instead of using network decompositions), our construction becomes much faster than all the previous constructions, and therefore provides the new state-of-the-art: 

\begin{theorem}\label{thm:NetDecomp} There is a deterministic algorithm that, in any $n$-node network, computes an $O(\log n)$-color $O(\log n \cdot \log\log\log n)$-strong-diameter decomposition in $\tO(\log^3 n)$ rounds\footnote{We use the notation $\tO(f(x)) = O(f(x) \cdot \poly\log f(x))$. } of the \congest model. The algorithm performs $\tO(m)$ computations in total, where $m$ denotes the number of edges.
\end{theorem}

\subsection{Hitting Set}
While network decomposition is a generic tool for derandomization in the \local model, and also a key tool for derandomization in the \congest model with extensive applications, a more basic tool that captures the usage of randomness in a range of distributed algorithms is \emph{hitting set}, as we describe next.

\paragraph{The Hitting Set Problem (basic case).} Given a collection of \lq\lq large\rq\rq\, sets in a ground set of elements, randomness gives us a very simple way of selecting a ``small\rq\rq\, portion of the elements such that we have at least one member of each set. The most basic variant is this: consider a bipartite graph $G=(A\sqcup B, E)$ where each node on one side $A$ has degree at least $k$. By using randomness, we can easily define a small subset $B'\subseteq B$ which, with high probability, has size $O(|B|\log n/k)$ and hits/dominates $A$, that is, each node $a\in A$ has a neighbor in $B'$. For that, simply include each element of $B$ in $B'$ with probability $p=O(\log n/k)$. This randomized selection in fact works in zero rounds. Finding such a small subset $B'$ in a deterministic manner is a key challenge in designing efficient deterministic distributed algorithms for many problems. For instance, Ghaffari and Kuhn~\cite{ghaffari2018congest-derandomizing} pointed out that this is the only use of randomness in some classic randomized algorithms for the construction of spanners and approximations of set cover. Indeed, a variant of this hitting set problem is a key ingredient even in our construction of the network decompositions mentioned above.

\paragraph{The Hitting Set Problem (general case).} Generalizing the problem allows us to capture a much wider range of applications. In some applications, we need to consider different sizes of the sets. Furthermore, we may not need to hit all sets, but instead, we would like to minimize the number, or the total cost, of those not hit. Following the bipartite graph terminology mentioned above, suppose each node $a\in A$ has a cost $c_a$, and its degree is denoted by $d_a$. Randomized selection with probability $p$ picks a subset $B'\subseteq B$ of size $p|B|$, in expectation, where the total cost of $A$-nodes that do not have a neighbor in $B'$ is $\sum_{a\in A} c_a (1-p)^{\deg(a)}$, in expectation. As a side comment, we note that in all applications that we are aware of, we may assume that $c_a\in [1, \poly(n)]$. Because of this, essentially without loss of generality, we can assume that for each node $a\in A$ we have $\deg(a)\leq O(\frac{1}{p} \cdot \log n)$. This is because the total expected cost of higher degree nodes is $1/\poly(n)$, which is negligible.

As an instructive example application, by defining $c_a:=\deg(a)$, we get that the total number of edges incident on $A$-nodes that are not hit is at most $O(|A|/p)$, in expectation. This particular guarantee is the sole application of randomness in some distributed constructions, e.g., the celebrated spanner construction of Baswana and Sen~\cite{baswana2007simple}. 

\paragraph{Prior deterministic distributed algorithms for hitting set.} There are two known distributed constructions for hitting set~\cite{ghaffari2018congest-derandomizing, bezdrighin2022deterministic}, as we review next. Both of these algorithms are based on showing that a small collection of random bits are sufficient for the randomized algorithm and then using the conditional expectation method to derandomize this. However, both algorithms are computationally inefficient and use superpolynomial-time computations. 

Ghaffari and Kuhn~\cite{ghaffari2018congest-derandomizing} observed that $O(\log n)$-wise independence is sufficient for the randomized algorithm in the basic hitting set problem, and thus $O(\log^2 n)$ bits of randomness are sufficient for the algorithm. Then, given a network decomposition with $c$ colors and strong diameter $d$, we can derandomize these bits one by one in a total of $O(cd \log^2 n)$ rounds, by processing the color classes one by one and fixing the bits in each color class in $O(d)$ time. However, the resulting algorithm is not computationally efficient: Each node has to perform $n^{O(\log n)}$-time local computations to calculate the conditional probabilities needed in the method of conditional expectation. 

Parter and Yogev~\cite{parter2018congested} pointed out that one can replace the $O(\log n)$-wise independence with a pseudorandomness generator for read-once DNFs and this reduces the number of bits to $O(\log n (\log\log n)^3)$--this was presented in a different context of spanners in the congested clique model of distributed computing. More recently, Bezdrighin et al.~\cite{bezdrighin2022deterministic} further reduced that bound to $O(\log n \log\log n)$, by applying a pseudorandomness generator of Gopalan and Yeudayoff~\cite{gopalan2020concentration}, which is particularly designed for hitting events. This decreased the round complexity slightly to $O(cd\log n\log\log n)$. However, the conditional probability computations still remain quite inefficient: they are $n^{O(\log\log n)}$-time, which is still super-polynomial. 

\paragraph{Our contribution.} 
Instead of viewing the randomized hitting set algorithm as a one-shot process, we turn it into a more gradual procedure. Concretely, we show that one can turn the natural randomized algorithm for the general hitting set problem into a number of randomized algorithms (bounded by $O(\log n)$), in such a way that pairwise independence is sufficient for analyzing each step. Because of this, we can derandomize each step separately (using an overall potential function that ensures that the result after derandomizing all steps has the same guarantees as discussed above for the randomized algorithm). Thanks to this, in contrast to the prior algorithms~\cite{ghaffari2018congest-derandomizing, bezdrighin2022deterministic} which required super-polynomial computations, our algorithm uses only $\tO(m)$ computations, summed up over the entire graph, where $m$ denotes the number of edges. Hence, our distributed algorithm directly provides a near-linear time low-depth deterministic parallel algorithm for the hitting set problem. 

\begin{theorem} [Informal] There is a deterministic distributed algorithm that in $\poly(\log n)$ rounds and using $\tO(m)$ total computations solves the generalized hitting set problem. That is, in the bipartite formulation mentioned above, the selected subset $B'$ has size $O(p \cdot |B|)$ and the total weight of nodes of $A$ not hit by $B'$ is $O(\sum_{a\in A} c_a (1-p)^{\deg(a)})$. 
\end{theorem}
We present the formal version of this theorem in \Cref{thm:hittingset-congest} for the \congest model of distributed computing, and in \Cref{thm:hittingset-pram} for the $\mathsf{PRAM}$ model of parallel computation.

\paragraph{Applications of hitting set.} This efficiently derandomized hitting set has significant applications for a number of graph problems of interest. In this paper, as two examples, we discuss spanners and distance oracles. In the case of spanners, this deterministic hitting set leads to the first deterministic spanner algorithm with the best-known stretch-size trade-off, polylogarithmic round complexity, that has near-linear time computations. The best previously known deterministic constructions required superpolynomial computations~\cite{bezdrighin2022deterministic} (and in the \cite{ghaffari2018congest-derandomizing} case, had extra logarithmic factors in size). The formal statements are as follows, and the proofs are presented in \Cref{subsec:spanners}.

\begin{corollary} (\textbf{general stretch spanners, unweighted and weighted})
There is deterministic distributed algorithm that, in $\poly(\log n)$ rounds of the \congest model and with total computations $\tO(m)$, for any integer $k\geq 1$, computes a $(2k-1)$-spanner with $O(nk + n^{1 + 1/k} \log k)$ and $O(nk + n^{1 + 1/k} k)$ edges for unweighted and weighted graphs, respectively.
\end{corollary}
\begin{corollary} (\textbf{ultra-sparse spanners})
There is deterministic distributed algorithm that, in $\poly(\log n)$ rounds of the \congest model and with total computations $\tO(m)$, computes a spanner with size $(1+o(1)) n$ and with stretch $\log n \cdot 2^{O(\log^* n)}$ in weighted graphs.
\end{corollary}

\medskip
By slight generalizations of our hitting set, we also obtain an efficient parallel derandomization of approximate distance oracles constructions: 
\begin{corollary} (\textbf{approximate distance oracle})
Given an undirected weighted graph $G = (V,E)$, a set of sources $S \subseteq V$ with $s=|S|$, and stretch and error parameters $k$ and $\eps > 0$, there is a deterministic algorithm that solves the source-restricted distance oracle problem with $\tO_{\eps}(ms^{1/k})$ work and $\tO_{\eps}(\poly(\log n))$ depth in the \pram model. The data structure has size $O(nk s^{1/k})$ and for each query $(u,v)$, the oracle can return a value $q$ in $O(k)$ time that satisfies
\begin{equation*}
    d(u,v) \leq q \leq (2k-1)(1 + \eps)d(u,v).
\end{equation*}
\end{corollary}

The proof is presented in \Cref{subsec:distance-oracles}. The corresponding centralized randomized construction was presented in the celebrated work of Thorup and Zwick~\cite{thorup2005approximate}. A centralized derandomization was given by Roddity, Thorup, and Zwick~\cite{roditty2005deterministic} but that approach does not appear to be applicable in parallel/distributed settings of computation.

	\section{Preliminaries}
	\label{sec:preliminaries}
	
	We use standard graph theoretic notation throughout the paper. 
	All graphs are undirected and unweighted. 
	For a graph $G = (V, E)$, we use $d_G$ or just $d$ to denote the distance metric induced by its edges. For sets of nodes $U,W \subseteq V(G)$, we generalize $d$ by $d(U, W) = \min_{u \in U, w \in W} d(u,w)$.

	\paragraph{Clustering}
	
	Given a graph $G$, its cluster $C$ is simply a subset of nodes of $V(G)$. 
	The strong-diameter $\diam(C)$ of a cluster $C$ is defined as $\diam(C) = \max_{u,v \in C} d_{G[C]}(u,v)$. 
	We note that there is a related notion of \emph{weak-diameter} of a cluster $C$ which is defined as the smallest $D$ such that $\forall u,v \in C : d_G(u,v) \le D$. That is, $C$ can even be disconnected, but there has to be a short path between any two nodes if we are allowed to use all nodes of $G$, not just nodes of $C$. 
	
	Although a cluster is simply a subset of $V(G)$, during the construction of a clustering, we keep its \emph{center} node $v \in C$ and often we work with an arbitrary breadth first search tree of $C$ from $v$. 
	
	The basic object we construct in this paper is \emph{separated clusterings}, which we formally define next.

	\begin{definition}[$s$-separated clustering]
	\label{def:clustering}
	Given an input graph $G$, a \emph{clustering} $\fC$ is a collection of disjoint \emph{clusters} $C_1, \dots, C_t$, such that for each $i$ we have $C_i \subseteq V(G)$. We say that the clustering has \emph{(strong-)diameter} $D$ whenever the diameter of each graph $G[C_i]$, $1 \le i \le t$, is at most $D$. We say that the clustering is $s$-separated if for every $1 \le i < j \le t$ we have $d_G(C_i, C_j) \ge s$. We sometime refer to this by saying that the clustering has separation $s$.  
	\end{definition}
	
	We will also need the following non-standard notion of $s$-hop degree of a cluster defined as follows. 
	\begin{definition}[$s$-hop degree]
	\label{def:shop_degree}
	Let $\fC$ be some clustering and $C \in \fC$ be a cluster with a fixed spanning tree $T_C$ rooted at $r \in C$. 
	The $s$-hop degree of $C$ in $\fC$ is the minimum number  $d$ such that for each $u \in C$ and the unique path $P_u$ from $u$ to $r$ in $T_C$ the following holds: The number of different clusters $C' \in \fC$ such that $d(P_u, C') \le s$ is at most $d$. 
	\end{definition}
	The $s$-hop degree of a clustering $\fC$ is the maximum $s$-hop degree over all clusters $C \in \fC$. 
	

	\section{Improved Network Decomposition, Outline}
	\label{sec:outline}
	To prove \Cref{thm:NetDecomp}, our core result is captured by the following low-diameter clustering statement, which clusters at least half of the vertices. \Cref{thm:NetDecomp} follows directly by repeating this clustering for $O(\log n)$ iterations, each time in the graph induced by the nodes that remain unclustered in the previous iterations. 

 \begin{theorem}\label{thm:Clustering-main} There is a deterministic \congest algorithm that runs in $\tO(\log^2 n)$ rounds and computes a clustering of at least $\frac{n}{2}$ nodes, with strong diameter $O(\log n \cdot \log\log\log n)$, and separation $2$.
 \end{theorem}

There are three ingredients in proving \Cref{thm:Clustering-main}, as we discuss next:

\paragraph{(A) Low-Degree Clustering.} The most important ingredient, captured by \Cref{thm:low_degree_main} and proven in \Cref{sec:low_degree}, is a clustering that manages to cluster half of the vertices but in which we have relaxed the separation/non-adjacency requirement of the clustering. Instead, we want each cluster to have $\separation$-hop degree of at most $\lceil 100 \log \log (n) \rceil$. See \cref{def:shop_degree} for the definition. 
For this ingredient, we present a randomized algorithm with pairwise analysis and then we derandomize it. 
	
	\begin{restatable*}{theorem}{lowdegreemain}
		\label{thm:low_degree_main}
		Let $\separation \geq 2$ be arbitrary. There exists a deterministic \congest algorithm running in $\tO(\separation \log^2(n))$ rounds which computes a clustering $\fC$ with
		
		\begin{enumerate}
			\item strong diameter $O(\separation \log(n))$,
			\item $\separation$-hop degree of at most $\lceil 100 \log \log (n) \rceil$, and
			\item the number of clustered nodes is at least $n/2$.
		\end{enumerate}
	\end{restatable*}

\paragraph{(B) From Low-Degree to Isolation.} The second ingredient, captured by \Cref{thm:subsampling_main} and proven in \Cref{sec:LowDegtoIsolation} is able to receive the clustering algorithm of (A) and turn it into a true clustering with separation $s$, but at the expense of reducing the number of clustered nodes by an $O(\log\log n)$ factor. For this ingredient as well, we first present a simple randomized algorithm with pairwise analysis, and then we derandomize it.

	\begin{restatable*}{theorem}{subsamplingmain}
		\label{thm:subsampling_main}
		Assume we are given a clustering $\fC$ with 
		
		\begin{enumerate}
			\item strong diameter $O(\separation \log (n))$ and
			\item $\separation$-hop degree of at most $\lceil 100 \log \log (n) \rceil$.
        \end{enumerate}
		There exists a deterministic \congest algorithm running in $\tO(\separation \log^2(n))$ rounds which computes a clustering $\fC^{out}$ with
		
		\begin{enumerate}
			\item strong diameter $O(\separation \log (n))$,
			\item separation of $\separation$ and
			\item the number of clustered nodes is $\frac{|\fC|}{1000 \log \log (n)}$.
		\end{enumerate} 
	\end{restatable*}
	
 \paragraph{(C) Improving Fraction of Clustering Nodes.} The third and last ingredient, captured by \Cref{thm:clusteringmorenodesmain} and proven in \Cref{sec:clusteringmorenodes}, receives the clustering algorithm of part (B) with a suitably high separation parameter (which is at least logarithmically related to the fraction of nodes clustered) and transforms it into a clustering of at least half of the nodes, at the expense of reducing the separation to simply $2$. This ingredient is a deterministic reduction and needs no derandomization and explains the final logarithm in the guarantees of \cref{thm:Clustering-main} (the first two logarithms are coming already from \cref{thm:low_degree_main}). 

\begin{restatable*}{theorem}{clusteringmorenodesmain}
    \label{thm:clusteringmorenodesmain}
	Let $x \geq 2$ be arbitrary.
	Assume there exists a deterministic \congest algorithm $\mathcal{A}$ running in $R$ rounds which computes a clustering $\fC$ with
	
	\begin{enumerate}
		\item strong diameter $O(x \log n)$,
		\item separation $10 \cdot x$ and
		\item clustering at least $\frac{n}{2^x}$ nodes.
	\end{enumerate}
	
	Then, there exists a deterministic \congest algorithm $\mathcal{A}'$ running in $O(2^x (R  + x \log n))$ rounds which computes a clustering $\fC'$ with
	
	\begin{enumerate}
		\item strong diameter $O(x \log n)$,
		\item separation $2$ and
		\item clustering at least $\frac{n}{2}$ nodes.
	\end{enumerate}
\end{restatable*}
	
	\vspace{1em}
Having all three ingredients \cref{thm:low_degree_main,thm:subsampling_main,thm:clusteringmorenodesmain}, we simply put them all together to prove \cref{thm:Clustering-main}. 

\begin{proof}[Proof of \Cref{thm:Clustering-main}]
Let $x=\lceil \log(2000 \log\log n)\rceil$. First, from \Cref{thm:low_degree_main}, we get a clustering of $n/2$ nodes with strong diameter $O(\log n\cdot \log\log\log n)$ and $10x$-hop degree at most $\lceil 100 \log \log (n) \rceil$, in $\tO(\log^2 n)$ rounds. Feeding this clustering algorithm to \Cref{thm:subsampling_main} produces a clustering algorithm that clusters $\frac{n}{2000 \log\log n}$ nodes with strong diameter $O(\log n\cdot \log\log\log n)$ and separation $10 x$, in $\tO(\log^2 n)$ rounds. Hence, this clustering can be put as input for \Cref{thm:clusteringmorenodesmain}, which as a result gives a clustering of at least $n/2$ nodes with strong diameter $O(\log n \cdot \log\log\log n)$, and separation $2$, in $\tO(\log^2 n)$ rounds.
\end{proof}

	\section{Low-Degree Clustering}

	\label{sec:low_degree}
	   
	   This section is devoted to proving the following theorem discussed in \cref{sec:outline}. 
	   
	\lowdegreemain
	
	\paragraph{Intuition Behind the Proof of \Cref{thm:low_degree_main}.} 
	
	In this paragraph, we give a brief intuition behind the proof of \cref{thm:low_degree_main}. %
	Our clustering algorithm can be viewed as derandomization of the randomized clustering algorithm of Miller, Peng, and Xu~\cite{miller2013parallel} (MPX). This is an algorithm that can cluster $n/2$ nodes with strong diameter $O(s \log n)$ such that the $s$-hop degree of the constructed clustering is in fact $1$, or in other words, the clustering is $s$-separated. 
	
	In the MPX algorithm, we simply run a breadth first search from all nodes of $V(G)$ at once, but every node starts the search only after a random delay computed as follows. Every node $v \in V(G)$ starts with the delay $\del(v) = O(s \log n)$. Next, each node starts flipping a coin and each time it comes up heads, it decreases its delay by $5s$. If it comes up tails, it stops the process. That is, the delays come from an exponential distribution; even more precisely, each node gets a head start coming from an exponential distribution, we talk about delays and add $O(s \log n)$ to make all numbers positive with high probability. 
	
	The guarantees of the MPX algorithm stem from the following observation. Let $u \in V(G)$ be arbitrary and let $\wait(u)$ be the first time $u$ is reached by above breadth first search with delays. Let $\frontier^{2s}(u)$ be the number of nodes $v \in V(G)$ such that $\del(v) + d(v,u) \le \wait(u) + 2s$. That is, $\frontier^{2s}(u)$ contains nodes who can reach $u$ after at most $2s$ additional steps after $u$ is reached for the first time. 
	We claim that with positive constant probability $\frontier^{2s}(u) = 1$, i.e., after the first node reaches $u$, it takes at least $2s$ additional steps until the next node reaches $u$. 
	
	To see this, replace each node $v \in V(G)$ by a runner on a real line who starts at position $d(u,v) + O(s \log n)$ (and may move toward left, as we soon discuss). Then, the exponential distribution that defines the delays corresponds to each runner flipping her coin until it comes up tails. For each heads, the runner runs distance $5s$ to the left. 
	We now let the runners flip the coins one by one. When a runner $r_j$ is flipping her coin, we consider the leftmost runner $r'_j$ out of the runners $r_1, \dots, r_{j-1}$ that already flipped their coins. 
	We observe that if $r_j$ at some point reaches a position at most $5s$ to the right from $r'_j$, we also have that $r_j$ runs to the distance $5s$ to the left of $r'_j$ with positive constant probability. 
	
	{\bf Derandomization: } Let us now explain the intuitive reason why we lose a factor of $O(\log \log n)$ in \cref{thm:low_degree_main}. Our derandomized algorithm simulates the coin flipping procedure step by step, for $O(\log n)$ steps, until every runner finally flips a tail and finishes. 
	In contrast to the previous simple algorithm, we now have to track our progress after every step. 
	So, our analysis is a derandomization of the following, different, randomized analysis of the same running process. 
	In this new randomized analysis, in each step $i$ and for each node $u$, we consider, very informally speaking, the event that the coin of all the runners that are currently at distance at most $2s$ from the leading runner comes up tail, where $t$ is a parameter we compute later. 
	The probability of this event is $2^{-t}$. 
	This means that the probability of this bad event happening in one of the $O(\log n)$ steps is at most $O(\log n) \cdot 2^{-t}$. 
	Choosing $t = O(\log \log n)$ makes this probability constant. Going back to the analysis of MPX, we get that at least half of nodes $u$ have $|\frontier^{2s}(u)| = O(\log  \log n)$. 
	
	Although this new randomized analysis loses a factor of $O(\log \log n)$, we can derandomize it in this section by setting up suitable potentials and derandomizing the coin flips of each step. To do so, we in fact simulate one fully-independent coin flip of each node in $O(\log\log n)$ steps where in each step we only use pairwise-independent random bits. 
	
    The rest of the section is structured as follows. In \cref{lem:pairwise_del_clustering}, we show how computing suitable delays gives rise to the final clustering. This step is simple and does not rely on derandomization. \cref{thm:low_degree_delay_main} then shows how to compute the node delays that simulate the MPX analysis as discussed above. 
    To sample even one ``coin flip'' of MPX, we need to invoke $O(\log\log n)$ times the local derandomization lemma of \cite{Faour2022}. One call of this lemma corresponds to \cref{thm:deterministic-goodSet-selection}.

	\paragraph{Basic Definitions}
	
		To prove \Cref{thm:low_degree_main}, we first need to define the notions of delay, waiting time, and a frontier:
	
	\begin{definition}[delay $\del$, waiting time $\wait_\del(u)$, and frontier $\frontier^{D}_{\del}(u)$]
	    \label{def:low_degree_del}
		A delay function $\del$ is a function assigning each node $u \in V$ a \emph{delay} $\del(u) \in \{0,1,\ldots,O(\separation \log (n))\}$.
		The waiting time of a node $u \in V$, with respect to a delay function $\del$,  is defined as
		\[\wait_{\del}(u) = \min_{v \in V} \left( \del(v) + d(v,u)\right).\] 
		The intuition behind $\wait_{\del}(u)$ is as follows: Assume that each node $v$ starts sending out a token at time $\del(v)$. Then, $\wait(u)$ is the time it takes until $u$ receives the first token.
		
		Furthermore, for every parameter $D \geq 0$, the \emph{frontier of width $D$} of a node $u \in V$, with respect to a delay function $\del$, is defined as
		
		\[\frontier^{D}_\del(u) = \{v \in V \colon \del(v) + d(v,u) \leq \wait_{\del}(u) + D\}.\] Informally, $\frontier^D(u)$ contains each node $v$ whose token arrives at $u$ at most $D$ time units after $u$ receives the first token.
	\end{definition}

 \smallskip

	 %
	   %
	  %
	  %

       \paragraph{Clustering from given delays.} The delay of each vertex is computed by a procedure provided in \Cref{alg:low_degree_delay_alg}. Before discussing that, we first explain how each delay function $\del$, along with a separation parameter $s$, give rise to a clustering $\fC^{\del}$: The clustering $\fC^{\del}$ clusters all the nodes that have a small frontier of width $2 \separation$. 
	  In particular, each node $u \in V(G)$ satisfying $|\frontier^{2\separation}(u)| \leq \lceil 100 \log \log(n)\rceil$ is included in some cluster of $\fC^{\del}$.
	  More concretely, each clustered node $u$ gets clustered to the cluster corresponding to the node with the smallest identifier in the set $\frontier^0(u)$. In other words, $u$ gets clustered with the cluster of the minimizer of $\wait(u)$, where we use the smallest identifier to break ties.  In the following text, we denote this node by $c_u$. See \cref{fig:mpx} for an illustration of this clustering. 
	  
	  \begin{figure}
	      \centering
	      \includegraphics[width = .6\textwidth]{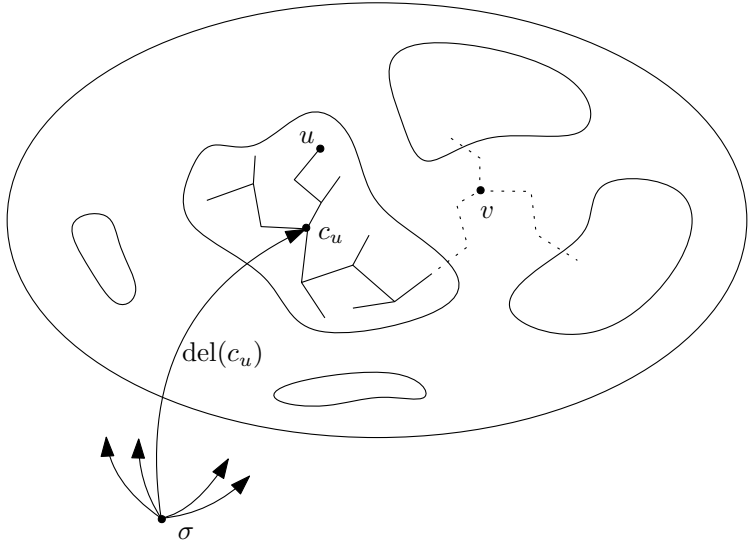}
	      \caption{The figure shows the run of the clustering algorithm constructing $\fC^{\del}$. 
	      The algorithm can be seen as starting a breadth first search from  a single node $\sigma$ connected to every node $u \in V(G)$ with an edge of length $\del(u)$ (the \congest implementation of the algorithm does not need to simulate any such node $\sigma$). 
	      The value $\wait(u)$ is the time until the search reaches the node $u$. The node that reaches $u$ the first is denote $c_u$. 
	      Moreover, we cluster only nodes such that the size of their frontier of width $2s$ is at most $O(\log\log n)$. For example, the node $v$ is not clustered because after it is reached by the first node $c_v$, it is reached by $\Omega(\log \log n)$ other nodes in the following $2s$ steps. 
	      It can be seen that for any $w$ on the path from $c_u$ to $u$, we have $\frontier^{2s}(w) \subseteq \frontier^{2s}(u)$, hence the constructed clusters are connected.  
	      }
	      \label{fig:mpx}
	      
	  \end{figure}

	  \begin{lemma}
	  	\label{lem:pairwise_del_clustering}
	  	Let $\del$ be a delay function and $\fC^{\del}$ the corresponding clustering, as described above. Then, the clustering $\fC^{\del}$ has
	  	\begin{enumerate}
	  		\item strong diameter $O(\separation \log n)$,
	  		\item $\separation$-hop degree at most $\lceil 100 \log \log(n)\rceil$ and
	  		\item the set of clustered nodes is equal to $|\{u \in V \colon \frontier^{2\separation}(u)| \leq \lceil 100 \log \log(n)\rceil\}|$.
	  	\end{enumerate}
	  	Moreover, the clustering $\fC^{\del}$ can be computed in $O(\separation \log n \log\log n)$ \congest rounds. 
	  \end{lemma}
	  
	  To prove \Cref{lem:pairwise_del_clustering}, we first observe that the frontiers have the following property. 
	  
	  \begin{claim}
	  \label{cl:frontier}
	  Let $w$ be any node on a shortest path from $u$ to $c_u$ and let $D \ge 0$. Then, we have (I) $\frontier^D(w) \subseteq \frontier^D(u)$, and (II) $c_w = c_u$.  
	  \end{claim}
	  \begin{proof}
	  First, we prove (I) $\frontier^D(w) \subseteq \frontier^D(u)$. Consider any $v \in \frontier^D(w)$. We prove that $v \in \frontier^D(u)$. Since $v \in \frontier^D(w)$, we have
	  	$$\del(v) + d(v, w) \le \del(c_u) + d(c_u, w) + D.$$ 
	  	Since $w$ lies on a shortest path from $c_u$ to $u$, we can add $d(w, u)$ to both sides of the equation to conclude that 
	  	$$\del(v) + d(v, w) + d(w, u) \le \del(c_u) + d(c_u, u) + D = \wait(u) + D,$$
	  	 and thus we have
	  	$$\del(v) + d(v, u) \le \wait(u) + D.$$
	  	Hence, we have $v \in \frontier^D(u)$ and (I) is proven. 

    Next, we prove (II) $c_w = c_u$.  In view of the above proof of (I), it suffices to show that $c_u \in \frontier^0(w)$. 
    To prove $c_u \in \frontier^0(w)$, we use the fact that $c_u \in \frontier^0(u)$ and write  %
	\[\del(c_u) + d(c_u,u) = \wait(u) \leq \wait(w) + d(w,u)\]
	%
	 Subtracting $d(w, u)$ from both sides of the equation and using that $w$ lies on a shortest path from $u$ to $c_u$ gives
	  	\[\del(c_u) + d(c_u,w) \leq \wait(w).\]
	  	Thus $c_u \in \frontier^0(w)$ and we are done. 
	  \end{proof}
	  \medskip
   
	  Having \Cref{cl:frontier}, we now go back to present a proof of \cref{lem:pairwise_del_clustering}.
	  \begin{proof}[Proof of \Cref{lem:pairwise_del_clustering}]
	  	We start with the first property. Let $u$ be an arbitrary clustered node and recall that $c_u$ is its cluster center.
	  	As $c_u \in \frontier^0(u)$, we have $d(c_u,u) = \wait(u)$.
	  	Moreover, 
	  	\[\wait(u) = \min_{v \in V} \del(v) + d(v,u) \leq \del(u) + d(u,u) = \del(u) = O(\separation \log n).\]
	  	Hence, we have $d(c_u,u) = O(\separation \log n)$. 
	  	Moreover, \cref{cl:frontier} gives that all nodes on a shortest path from $c_u$ to $u$ are also clustered to $c_u$, implying that the diameter of the cluster is $O(s \log n)$.

	  	Next, we prove the second property.
	  	Consider an arbitrary clustered node $w$. We first show that for an arbitrary clustered node $y$ with $d(w,y) \leq \separation$, it holds that $c_y \in \frontier^{2\separation}(w)$.
	  	To see this, we first use the definition of $c_w$ to write
	  	\begin{align*}
	  	    \del(c_y) + d(c_y,y)  
	  		\leq \del(c_w) + d(c_w,y)
	  	\end{align*}
	  	On one hand, we can use triangle inequality to lower bound the left-hand side by
	  	\begin{align*}
	  		  	    \del(c_y) + d(c_y,y) \ge \del(c_y) + d(c_y, w) - d(w,y) 
\end{align*}
On the other hand, we can use triangle inequality to upper bound the right hand side by
\begin{align*}
	  	    \del(c_w) + d(c_w,y)
	  	    \le \del(c_w) + d(c_w, w) + d(w, y). 
	  	\end{align*}
	  	Putting the two bounds together, we conclude that 
	  	\begin{align*}
	  	    \del(c_y) + d(c_y, w) \le \del(c_w) + d(c_w, w) + 2d(w,y) \le \wait(w) + 2s, 
	  	\end{align*}
	  	where we used our assumption $d(w,y) \le 2s$. Thus, $c_y \in \frontier^{2\separation}(w)$.
	  	
	  	Now, let $u$ be an arbitrary clustered node and $P_u$ the unique path between $u$ and $c_u$ in the tree associated with the cluster. Furthermore, let $C \in \fC^{\del}$ be a cluster with $d(P_u,C) \leq \separation$. Then there exists $w \in P_u$ and $y \in C$ with $d(w,y) \leq \separation$ and the discussion above implies $c_y \in \frontier^{2\separation} (w)$. 
	  	
        We now use \cref{cl:frontier}, which implies that $\frontier^{2\separation} (w) \subseteq \frontier^{2\separation}(u)$. 
	  	Hence, for each cluster $C$ with $d(P_u,C) \leq \separation$, the corresponding cluster center is contained in $\frontier^{2\separation}(u)$. As $u$ is clustered, we know that $|\frontier^{2\separation}(u)| \leq \lceil100 \log \log n\rceil$. Therefore, the $\separation$-hop degree of $\fC^{\del}$ is at most $\lceil100 \log \log (n)\rceil$.
	  
	  The third property follows directly from the definition. 	
	  To finish the proof, we need to show that the algorithm can be implemented in $O(s \log n \cdot \log\log n)$ rounds. 
	  To see this, note that we can compute for each node $u$ whether $|\frontier^{2s}(u)| \le \lceil 100 \log \log(n)\rceil$ or not, as follows: We run a variant of breadth first search that takes into account the delays, where each node $v$ starts sending out a BFS token at time $\del(v)$. Recall that in classical breadth first search, after a node $u$ is reached for the first time by a token (sent from $c_u$), it broadcasts this token to all its neighbors and then it does not redirect any other tokens sent to it. 
    In our version of the search, each node stops redirecting only after at least $\lceil 100 \log \log(n)\rceil$ tokens arrived (we do not take into account tokens that have already arrived earlier) or after it counts $2s$ steps from the arrival of the first token. 
    It can be seen that this algorithm can be implemented in the desired number of rounds. Moreover, every node $u$ learns the value of $|\frontier^{2s}(u)|$ whenever the value is at most $\lceil 100 \log \log(n)\rceil$ and otherwise, it learns the value is larger than this threshold. 
	  \end{proof}

	  In view of \cref{lem:pairwise_del_clustering}, to prove the randomized variant of \Cref{thm:low_degree_main} with pairwise analysis, it suffices to show that \cref{alg:low_degree_delay_alg} stated next computes a delay function $\del$ such that the expected number of nodes $u$ with $|\frontier^{2\separation}(u)| \leq \lceil 100 \log \log(n)\rceil$ is at least $n/2$. We later discuss how this is derandomized.

	\subsection{Computing Delays}

	This subsection is dedicated to proving the following theorem that asserts that we can compute a suitable delay function that can be plugged in \cref{lem:pairwise_del_clustering} that constructs a clustering from it. 
	
	\begin{theorem}
	    \label{thm:low_degree_delay_main}
        \cref{alg:low_degree_delay_alg} runs in $\tO(\separation \log^2(n))$ \congest rounds and computes a delay function $\del$ that satisfies 
        \begin{align}
        \label{eq:houkacka}
        |\{u \in V \colon |\frontier_{\del}^{2s}(u)| \leq \lceil 100 \log \log (n) \rceil\}| \geq n/2.    
        \end{align}
         
    \end{theorem}

	\begin{restatable}{algorithm}{delays}
	\caption{Computing Delay Function $\del$}
	\label{alg:low_degree_delay_alg}
	Input: A parameter $s$, an algorithm $\fA_{i,j}$ computing a \emph{good set} from \cref{def:low_degree_del_good_set} in $\tO(s \log n)$ rounds\\
	Output: A delay function $\del$ from \cref{def:low_degree_del} satisfying \cref{eq:houkacka}
	\begin{algorithmic}[1] 
		\Procedure{Delays}{} 
		\State $V_0^ { \act} \leftarrow V$
		\State $R \leftarrow \lfloor 2 \log(n) \rfloor$
		\State $k \leftarrow \lceil 100 \log \log (n) \rceil$ 
		\State $\forall u \in V \colon \del_0(u) \leftarrow 5 \separation R$
		\For{$i \leftarrow 1,2, \ldots, R$}
		    \State $W_{i,0} \leftarrow \emptyset$
		    \For{$j \leftarrow 1,2,\ldots,k$}
		        \State  $S_{i,j} \leftarrow  \mathcal{A}_{i,j}(\del_{i-1}, W_{i,j-1})$ \Comment{$S_{i,j} \subseteq V_{i-1}^ { \act}$}  
		        \State $W_{i,j} \leftarrow W_{i,j-1} \cup S_{i,j}$ 
		       
		    \EndFor
			\State $V_i^{ \act} \leftarrow W_{i,k}$
			\For{$\forall u \in V$}
			    \If{$u \in V^\act_i$}
			        \State $\del_i(u) \leftarrow \del_{i-1}(u) - 5s$
			    \Else
			        \State $\del_i(u) \leftarrow \del_{i-1}(u)$
                \EndIf			    
			\EndFor
		\EndFor
		\State $\del \leftarrow \del_R$
		\State \Return $\del$
		\EndProcedure
	\end{algorithmic}
\end{restatable}

\paragraph{Intuitive Description of \cref{alg:low_degree_delay_alg}.} 
The algorithm runs in $R = \lfloor 2 \log(n)\rfloor$ phases and each phase consists of $k = \lceil 100 \log \log n\rceil$ iterations.
In iteration $j$ of phase $i$, algorithm $\mathcal{A}_{i,j}$ is a deterministic algorithm which computes a good set $\mathcal{S}_{i,j} \subseteq V^{\act}_{i-1}$ as defined later in \cref{def:low_degree_del_good_set}. The algorithm description of $\mathcal{A}_{i,j}$ is deferred to \cref{sec:low_degree_local_derandomization}. The high-level intuition is that $\mathcal{A}_{i,j}$ derandomizes the randomized process which obtains $S_{i,j}$ from $V^{\act}_{i-1}$ by including each vertex with probability $\frac{1}{4k}$, pairwise independently.
Repeating this pairwise independent sampling process $k$ times then simulates including each vertex from $V^{\act}_{i-1}$ to $V^{\act}_{i}$ with positive probability. 
The derandomization of the pairwise independent process is done efficiently using a novel local derandomization procedure introduced in \cite{Faour2022} which essentially allows to efficiently derandomize algorithms that only rely on pairwise analysis.

 Throughout the algorithm, each node is assigned a delay. At the beginning, each node $u$ is assigned a delay of $\del_0(u) = 5 \separation R = O(\separation \log n)$. In each subsequent phase, for each node $u$, we have two possibilities: if $u \in V^{\act}_i$, the delay of node $u$ is decreased by $5 \separation$, i.e., $\del_i(u) = \del_{i-1}(u) - 5 \separation$ if $u \in V^{\act}_i$; if $u \notin V^{\act}_i$, then its delay stays the same, i.e., $\del_i(u) = \del_{i-1}(u)$ if $u \notin V^{\act}_i$. 

For every $u \in V$, we define the shorthand $\wait_i(u) = \wait_{\del_i}(u)$ and for every $D \geq 0$, we define $\frontier_i^D(u) = \frontier_{\del_i}^D(u)$.

	\paragraph{Communication Primitives.} 
    For the deterministic algorithm $\mathcal{A}_{i,j}$ which computes the set $S_{i,j}$, it is important that each node $u$ can efficiently compute the set $\alive_{i-1}(u)$ and $\dead_{i-1}(u)$ that are defined next. 
    Let us give a brief intuition behind the definition. In the ``runner intuition'' from the beginning of the section, we want to know in every step all runners that are currently at distance at most $2s$ after the front runner. For these runners, we want to ensure that not all of them stop running in one step. In the reality of the distributed \congest model, we however cannot compute even the size of $\frontier^{2s}(u)$. 
    
    Fortunately, for our purposes if the number of ``runners'' that are distance at most $2s$ from the front runner is larger than $O(\log\log n)$, it roughly speaking suffices to work with the first $O(\log \log n)$ runners in the analysis. This is formalized by the following definition of alive and dead nodes (dead nodes are runners that stopped flipping coins).  
    
\begin{restatable}{definition}{alivedead}[$\alive_i(u)$/$\dead_i(u)$] 
For every vertex $u \in V$ and $i \in \{0,1,\ldots,R\}$, let $\dead_i(u) \subseteq \frontier^{2 \separation}_i(u) \setminus V_i^{\act}$ be an arbitrary subset of size $\min(k, |\frontier^{2 \separation}_i(u) \setminus V_i^{\act}|)$ and $\alive_i(u) \subseteq \frontier^{2 \separation}_i(u) \cap V_i^{\act}$ be an arbitrary subset of size $\min(k - |\dead_i(u)|, |\frontier^{2 \separation}_i(u) \cap V_i^{\act}|)$. 
\end{restatable}
Note that $|\alive_i(u)| + |\dead_i(u)| \leq \min(k, |\frontier^{2 \separation}_i(u)|)$ and $|\dead_i(u)| \leq |\alive_{i-1}(u)| + |\dead_{i-1}(u)|$.
For each node $v \in V$, let $M_{i-1}(v) = \{u \in V \colon v \in \alive_{i-1}(u)\}$. Then, we need some simultaneous and efficient communication, that allows each $v\in V$ to broadcast a message to all nodes in $M_{i-1}(v)$, and for $v$ to receive an aggregate of messages prepared for $v$ in nodes $M_{i-1}(v)$.

    \begin{lemma}
      \label{lem:low_degree_delays_communication}
      Suppose that we are at the beginning of some phase $i \in [R]$. Given delay function $\del_{i-1}$, and given the set $V_{i-1}^{\act}$, there exists a \congest algorithm running in $\tO(\separation \log n)$ rounds which computes for each node $u \in V$ the sets $\alive_{i-1}(u)$ and $\dead_i(u)$. Moreover, let $M_{i-1}(v) = \{u \in V \colon v \in \alive_{i-1}(u)\}.$ Then, there exists an $\tO(\separation \log n)$ round \congest algorithm that allows each node $v$ to send one $O(\log n)$-bit message that is delivered to all nodes in $M_{i-1}(v)$. Similarly, there also exists an $\tO(\separation \log n)$ round \congest algorithm that given $O(\log n)$-bit messages prepared at nodes in $M_{i-1}(v)$ specific for node $v$, it allows node $v$ to receive an aggregation of these messages, e.g., the summation of the values, in $\tO(\separation \log n)$ rounds. 
    \end{lemma}
    \begin{proof}[Proof of \Cref{lem:low_degree_delays_communication}] We run a variant of breadth first search (BFS) that takes into account the delays, and runs in $\tO(s\log n)$ rounds: Each node $v$ starts sending out a BFS token at time $\del_{i-1}(v)$, where the token also includes the information whether $v\in V_{i}^{\act}$ or not. During the entire process, each node $u$ forwards per time step at most $k= \lceil100\log\log n \rceil$ BFS tokens, breaking ties in favoring of tokens coming from nodes $v\notin V_{i}^{\act}$. That is, all tokens that arrive at the same time step are forwarded in the next time step, except that the node forwards at most $k$ tokens in this time step, and moreover, the node first includes all tokens from nodes $v\notin V_{i}^{\act}$ (up to $k$) before including tokens from nodes $v\in V_{i}^{\act}$. Since per time step each node forwards at most $k$ tokens, each time step can be implemented in at most $k$ rounds of the \congest model. Furthermore, node $u$ starts counting time from the moment that it received the very first token (while forwarding any received tokens, up to $k$ per time step), and after $2s$ time steps have passed, node $u$ does not forward any other tokens. 
    
    Let us think of one tree for each node $v$, which is rooted at $v$ and includes all nodes $u$ that have received the token of node $v$. Every node $u$ receives the tokens of all nodes in $\frontier_{i-1}^{2s}(u)$, if there are at most $k$ of them. If there are more than $k$, node $u$ learns at least $k$ of them, with the following guarantee: The set of learned tokens includes all tokens from $v\notin V_{i}^{\act}$, up to $k$ (if there were more).
    
    Hence, given the received tokens, each node $u$ can form $\dead_{i-1}(u) \subseteq \frontier^{2 \separation}_{i-1}(u) \setminus V_{i-1}^{\act}$, which is subset of size $\min(k, |\frontier^{2 \separation}_{i-1}(u) \setminus V_{i-1}^{\act}|)$. Furthermore, node $u$ can form $\alive_{i-1}(u) \subseteq \frontier^{2 \separation}_{i-1}(u) \cap V_{i-1}^{\act}$, which is a subset of size $\min(k - |\dead_{i-1}(u)|, |\frontier^{2 \separation}_{i-1}(u) \cap V_{i-1}^{\act}|)$.
    
    By repeating the same communication, each node $v$ is able to send one message which is delivered to all nodes $M_{i-1}(v) = \{u \in V \colon v \in \alive_{i-1}(u)\}$, all simultaneously in $\tO(s\log n)$ rounds. Moreover, by repeating the same communication pattern but in the reverse direction of time, we can do an aggregation along each tree, again all simultaneously in $\tO(s\log n)$ rounds, allowing each node $v$ to receive an aggregation of the messages prepared for $v$ in nodes $M_{i-1}(v)$.
    \end{proof}

\paragraph{Potential Functions}

In this paragraph, we define an outer potential $\Phi_i$ for every phase $i$ and an inner potential $\phi_{i,j}$ for every
iteration $j$ within phase $i$. The inner potential satisfies that if $\phi_{i,j-1} \leq \phi_{i,j}$ in each iteration $j$, then $\Phi_i \leq \Phi_{i-1} + n$. The outer potential satisfies that $\Phi_0 = 2n$ and if $\Phi_R \leq 10n \log(n)$, then $|\{u \in V \colon |\frontier^{2 \separation}(u)|\}| \geq \frac{9n}{10}$.

\begin{definition}[Outer Potential]
\label{def:low_degree_delays_outer_potential}

For every $i \in \{0,1,\ldots,R\}$, the outer potential of a node $u$ after phase $i$ is defined as

\[\Phi_i(u) = e^{\frac{|\dead_i(u)|}{10}}.\]
The outer potential after phase $i$ is defined as

\[\Phi_i = \sum_{u \in V} \Phi_i(u) + 2^i|V^{\act}_i|.\]
\end{definition}
    Here, "after phase $0$" should be read as "the beginning of phase $1$".
    \cref{alg:low_degree_delay_alg} will make sure that the outer potential is sufficiently small. A small outer potential after phase $i$ implies on one hand that there are not too many nodes $u$ for which $|\dead_i(u)|$ is large and on the other hand ensures that there are not too many nodes in $V^{\act}_i$, i.e., $|V^{\act}_i| \lesssim\frac{n}{2^i}$.
    
    The following lemma captures the usefulness of the outer potential.

	\begin{lemma}[Outer Potential Lemma]
		\label{lem:low_degree_delays_outer_potential_lemma}
		We have $\Phi_0 \leq 2n$. Moreover, if $\Phi_R \leq 10 n \log(n)$, then $|\{u \in V \colon |\frontier_{\del}^{2\separation}(u)| \leq 100 \log \log (n)\}| \geq \frac{9n}{10}$.
	\end{lemma}
	\begin{proof}
		First, note that $\Phi_R \geq 2^R|V_R^ { \act}| > 10n \log(n)|V_R^{ \act}|$.
		As we assume that $\Phi_R \leq 10 n \log(n)$, this directly implies $V_R^ { \act} = \emptyset$.
		In particular, every $u \in V$ with $|\frontier^{2\separation}_{\del}(u)| > 100 \log \log (n)$ contributes
		\[\Phi_R(u) = e^{\frac{|\dead_R(u)|}{10}}\geq e^{\frac{\min(k,|\frontier_{\del}^{2\separation}(u)|}{10})} \geq  100 \log(n) \]
		to the potential. Hence, there can be at most $\Phi_R/(100 \log(n)) \leq n/10$ such nodes and therefore $\{u \in V \colon |\frontier_{\del}^{2\separation}(u)| \leq 100 \log \log (n)\} \geq \frac{9n}{10}$, as desired.
	\end{proof}
	
	\begin{restatable}{definition}{pessimisticprobability}[Pessimistic Estimator Probability $p_{i,j}(u)$]
	For $i \in [R]$ and $j \in \{0,1,\ldots,k\}$, the pessimistic estimator probability of a node $u$ after iteration $j$ within phase $i$ is defined as
\[p_{i,j}(u) = I(\alive_{i-1}(u) \cap W_{i,j} = \emptyset) \cdot \left(1 - \frac{|\alive_{i-1}(u)|}{10k}\right)^{k-j}.\]
\end{restatable}
Here, "after iteration $0$", should be read as "the beginning of iteration $1$".
Let us briefly elaborate on the definition of $p_{i,j}(u)$. Assume that we would compute $S_{i,j}$ by sampling each vertex in $V^{\act}_{i-1}$ with probability $\frac{1}{4k}$ pairwise independently.
By a simple pairwise analysis, one can show that this would imply $\Pr[S_{i,j} \cap \alive_{i-1}(u) \neq \emptyset] \geq \frac{|\alive_{i-1}(u)|}{10k}$. Hence, if we are currently at the beginning of iteration $j$ within phase $i$ just prior to sampling the set $S_{i,j}$, then $p_{i,j-1}(u)$ is an upper bound on the probability that no node in $\alive_{i-1}(u)$ is contained in $ V^{\act}_i$ (which one should think of as a bad event). 
	
	\begin{restatable}{definition}{innerpotential}[Inner Potential]
	\label{def:low_degree_del_inner}

    The inner potential of a node $u$ after iteration $j$ of phase $i$ is defined as

\[\phi_{i,j}(u) = p_{i,j}(u) e^{\frac{|\dead_{i-1}(u)| + |\alive_{i-1}(u)|}{10}}.\]
The inner potential after iteration $j$ of phase $i$ is defined as
\[\phi_{i,j} = \sum_{u \in V} \phi_{i,j}(u) + |W_{i,j}|2^i + \frac{k-j}{k}2^{i-1} |V_{i-1}^{\act}|.\]
\end{restatable}

Again, assume for a moment that we would compute $S_{i,j}$ by sampling each vertex in $V^{\act}_{i-1}$ with probability $\frac{1}{4k}$ pairwise independently. Assume we are at the beginning of iteration $j$ within phase $i$ just prior to sampling $S_{i,j}$. Then, using the fact that $\E[p_{i,j}(u)] \leq p_{i,j-1}(u)$, one directly gets that $\E[\phi_{i,j}(u)] \leq \phi_{i,j}(u)$ and it also follows that $\E[\phi_{i,j}] \leq \phi_{i,j-1}$. Moreover, one can also show that  $\E[\Phi_i(u)] \leq \phi_{i,j-1}(u) + 1$ and $\E[\Phi_i] \leq \phi_{i,j-1} + n$.
In more detail, if at least one node in $\alive_{i-1}(u)$ is contained in $V^{\act}_i$, one can show that this implies $\dead_i(u) = \emptyset$ and therefore $\Phi_i(u) = 1$. On the other hand, in the previous discussion we mentioned that with probability at most $p_{i,j-1}(u)$ no node in $\alive_{i-1}(u)$ is included in $V^{\act}_i$, and as $|\dead_{i-1}(u)| + |\alive_{i-1}(u)| \leq |\dead_i(u)|$, we have $\ p_{i,j-1}(u)\Phi_i(u) \leq \phi_{i,j-1}$.

\begin{lemma}[Inner Potential Lemma]
\label{lem:low_degree_delays_inner}
For $i \in [R]$ and $j \in \{0,1,\ldots,k\}$,
Assume that in every iteration $j$ of phase $i$, $S_{i,j}$ is computed in such a way that $\phi_{i,j} \leq \phi_{i,j-1}$. Then, $\Phi_i \leq \Phi_{i-1} + n$ and $\Phi_R \leq 4n \log(n)$.

\end{lemma}

\begin{proof}
For each node $u \in V$, we have

    \[\phi_{i,0}(u) = p_{i,0}(u)e^{\frac{|\dead_{i-1}(u)| + |\alive_{i-1}(u)|}{10}} = \left( 1 - \frac{|\alive_{i-1}(u)|}{10k}\right)^{k-0} e^{\frac{|\dead_{i-1}(u)| + |\alive_{i-1}(u)|}{10}} \leq e^{\frac{|\dead_{i-1}(u)|}{10}} = \Phi_{i-1}(u).\]
Therefore,

\[\phi_{i,0} = \sum_{u \in V} \phi_{i,0}(u) + |W_{i,0}|2^i + \frac{k-0}{k}2^{i-1} |V^{\act}_{i-1}| \leq \sum_{u \in V} \Phi_{i-1}(u)   + 2^{i-1}|V^{\act}_{i-1}| = \Phi_{i-1}.\]

Consider an arbitrary $u \in V$. Next, we show that

\[e^{\frac{|\dead_i(u)|}{10k}}=: \Phi_i(u) \leq \phi_{i,k}(u) + 1 = I(\alive_{i-1}(u) \cap V_i^\act = \emptyset)e^{\frac{|\dead_{i-1}(u)| + |\alive_{i-1}(u)|}{10k}} + 1.\]

    It is easy to verify that the inequality is satisfied if $\alive_{i-1}(u) \cap V_i^\act = \emptyset$, as $|\dead_i(u)| \leq |\dead_{i-1}(u)| + |\alive_{i-1}(u)|$.
	Therefore, it remains to consider the case that there exists at least one node $v \in \alive_{i-1}(u) \cap V_i^{\act}$. The existence of such a node $v$ implies
	\[\wait_i(u) \leq  \del_i(v) + d(v,u) =  \del_{i-1}(v) + d(v,u) - 5 \separation \leq \wait_{i-1}(u) + 2 \separation - 5 \separation = \wait_{i-1}(u) - 3 \separation.\]
	For every node $w \notin V^{\act}_i$, we have $ \del_i(w) =  \del_{i-1}(w)$ and therefore
	\[ \del_i(w) + d(w,u) =  \del_{i-1}(w) + d(w,u) \geq \wait_{i-1}(u) \geq \wait_i(u) + 3 \separation > \wait_i(u) + 2 \separation\]
	and thus $w \notin \frontier_i^{2 \separation}(u)$.
	Hence, $\dead_i(u) = \emptyset$ and the inequality is satisfied.
Therefore,

\[\phi_{i,k} = \sum_{u \in V} \phi_{i,k}(u) + |W_{i,k}|2^i + \frac{k-k}{k}2^{i-1} |V^{\act}_{i-1}| \leq \left(\sum_{u \in V} \Phi_i(u) - 1 \right)   + |V^{\act}_i|2^i = \Phi_i - n.\]

A simple induction implies $\phi_{i,k} \leq \phi_{i,0}$. 
Therefore,
\[\Phi_i \leq \phi_{i,k} + n \leq \phi_{i,0} + n = \Phi_{i-1} + n.\]
As $\Phi_0 \leq 2n$ according to \cref{lem:low_degree_delays_outer_potential_lemma}, a simple induction implies
\[\Phi_R \leq \Phi_0 + Rn \leq (2+R)n \leq 4n \log n.\]
\end{proof}

	\paragraph{Good Set $S_{i,j}$:}
	
	We are now going to define the good set of nodes $S_{i,j}$. Note that this is the part of \cref{alg:low_degree_delay_alg} whose definition we postponed. 
	
	\begin{definition}[Good Set $S_{i,j}$]
	    \label{def:low_degree_del_good_set}
	    For a set $S_{i,j} \subseteq V^{\act}_{i-1}$  and $u \in V$, let
	    
	    \[Y_{i,j}(u) = 1 - |\alive_{i-1}(u) \cap S_{i,j}| + \binom{|\alive_{i-1}(u) \cap S_{i,j}|}{2}.\]
	    
	    We refer to the set $S_{i,j}$ as good if 
	    
	    \[  \sum_{u \in V} Y_{i,j}(u) \frac{\phi_{i,j-1}(u)}{1 - (|\alive_{i-1}(u)|/(10k))} + |S_{i,j}|\cdot 2^i \leq \sum_{u \in V} \phi_{i,j-1}(u) + \frac{2^{i-1}}{k}|V^{\act}_{i-1}|.\]
	\end{definition}
	
	\begin{lemma}
	\label{lem:low_degree_del_potential_nonicreasing}
	If $S_{i,j}$ is a good set, then $\phi_{i,j} \leq \phi_{i,j-1}$.
	\end{lemma}
	\begin{proof}
	
	For each $u \in V$, we have
	\begin{align*}
	\frac{I(\alive_{i-1}(u) \cap S_{i,j} = \emptyset)p_{i,j-1}(u)}{1 - (|\alive_{i-1}(u)|/(10k))} &= 
	    \frac{I(\alive_{i-1}(u) \cap S_{i,j} = \emptyset)I(\alive_{i-1}(u) \cap W_{i,j-1} = \emptyset) \cdot \left(1 - \frac{|\alive_{i-1}(u)|}{10k}\right)^{k-(j-1)}}{1 - (|\alive_{i-1}(u)|/(10k))} \\
	    &= I(\alive_{i-1}(u) \cap W_{i,j} = \emptyset) \left(1 - \frac{|\alive_{i-1}(u)|}{10k}\right)^{k-j} \\
	    &= p_{i,j}(u).
	\end{align*}
	
	It also holds that $I(\alive_{i-1}(u) \cap S_{i,j} = \emptyset) \leq Y_{i,j}(u)$. Therefore,

\begin{align*}
    Y_{i,j}(u)  \frac{\phi_{i,j-1}(u)}{1 - (|\alive_{i-1}(u)|/(10k))} &\geq I(\alive_{i-1}(u) \cap S_{i,j} = \emptyset)\frac{p_{i,j-1}(u)e^{\frac{|\dead_{i-1}(u)| + |\alive_{i-1}(u)|}{10}}}{1 - (|\alive_{i-1}(u)|/10k)} \\
    &= p_{i,j}(u) e^{\frac{|\dead_{i-1}(u)| + |\alive_{i-1}(u)|}{10}} \\
    &= \phi_{i,j}(u).
\end{align*}

Thus, we get

\begin{align*}
    \phi_{i,j} &= \sum_{u \in V} \phi_{i,j}(u) + |W_{i,j}|2^i + \frac{k-j}{k}2^{i-1}|V_{i-1}^{\act}| \\
    &\leq \sum_{u  \in V} Y_{i,j}(u)  \frac{\phi_{i,j-1}(u)}{1 - (|\alive_i(u)|/10k)} + (|S_{i,j}| + |W_{i,j-1}|)2^i + \frac{k-j}{k}2^i|V^{\act}_{i-1}| \\
    &\leq \sum_{u \in V} \phi_{i,j-1}(u) + \frac{2^{i-1}}{k}|V^{\act}_{i-1}| + |W_{i,j-1}|2^i + \frac{k-j}{k}2^{i-1}|V^{\act}_{i-1}| \\
    &= \phi_{i,j-1}.
\end{align*}
	\end{proof}
	
	We now combine all the pieces to prove the main theorem of this subsection.
	
	\paragraph{Proof of \cref{thm:low_degree_delay_main}}
	
	We assume that in iteration $j$ of phase $i$, $\mathcal{A}_{i,j}$ computes a good set $S_{i,j}$. Therefore, \cref{lem:low_degree_del_potential_nonicreasing} implies that $\phi_{i,j} \leq \phi_{i,j-1}$. According to \cref{lem:low_degree_delays_inner}, this implies that $\Phi_R \leq 4n \log(n)$. Therefore, \cref{lem:low_degree_delays_outer_potential_lemma} implies that $|\{u \in V \colon |\frontier^{2 \separation}_{\del}(u)| \leq \lceil 100 \log \log (n)\rceil\}| \geq n/2$, as desired. It remains to discuss the \congest round complexity.
	
	\cref{alg:low_degree_delay_alg} has $R \cdot k = \tO(\log n)$ iterations in total. In iteration $i$ of phase $j$, algorithm $\mathcal{A}_{i,j}$ runs in $\tO(\separation \log n)$ \congest rounds. Hence, the overall \congest complexity of \cref{alg:low_degree_delay_alg} is $\tO(\separation \log^2 n)$.
	
	\paragraph{Global Derandomization}
	
	Here, we informally sketch a variant of \cref{alg:low_degree_delay_alg} which performs a global derandomization using the method of conditional expectation. A more formal discussion of this approach, though in a different context, is discussed in \cref{subsec:hittingset-impl} where we derandomize our algorithm for the hitting set problem in the \congest model. See in particular \cref{thm:hittingset-congest}.
	
	\begin{definition}[Good Random Set $S_{i,j}$ (In Expectation)]
	    \label{def:low_degree_del_random_good_set}
	    For a set $S_{i,j} \subseteq V^{\act}_{i-1}$  and $u \in V$, let
	    
	    \[Y_{i,j}(u) = 1 - |\alive_{i-1}(u) \cap S_{i,j}| + \binom{|\alive_{i-1}(u) \cap S_{i,j}|}{2}.\]
	    We refer to a randomly computed subset $S_{i,j} \subseteq V^{\act}_{i-1}$ as good in expectation if 
	    \[  \E \left[\sum_{u \in V} Y_{i,j}(u) \frac{\phi_{i,j-1}(u)}{1 - (|\alive_{i-1}(u)|/(10k))} + |S_{i,j}|\cdot 2^i \right] \leq \sum_{u \in V} \phi_{i,j-1}(u) + \frac{2^{i-1}}{k}|V^{\act}_{i-1}|.\]
	\end{definition}
	Note that we can recover \cref{def:low_degree_del_good_set} if we drop the expectation. Assume we choose $S_{i,j}$ by including each node in $V^{\act}_{i-1}$ with probability $\frac{1}{4k}$, pairwise independently. One can show that the resulting set $S_{i,j}$ is good in expectation. Moreover, the pairwise distribution over the random set $S_{i,j}$ can be realized with a random seed length of $\tO(\log n)$ using the construction of~\cite{roditty2005deterministic, luby1993removing} that is described in \Cref{subsec:hittingset-impl}.
	The goal is now to fix the random seed one by one in such a way that the resulting deterministic set $S_{i,j}$ is a good set.
	
	For the following discussion, let $X = \sum_{u \in V} Y_{i,j}(u) \frac{\phi_{i,j-1}(u)}{1 - (|\alive_{i-1}(u)|/(10k))} + |S_{i,j}|\cdot 2^i$.
	The method of conditional expectation works by fixing the bits of the random seed one by one, each time fixing the $i$-th bit in such a way that 
	\[\E[X|\text{first $i$ bits are fixed to $b_0,\ldots,b_i$}] \leq \E[X|\text{first $i-1$ bits are fixed to $b_0,\ldots,b_{i-1}$}].\]
	In particular, this ensures that 
	\[\E[X|\text{all bits are fixed}] \leq \E[X] \leq \sum_{u \in V} \phi_{i,j-1}(u) + \frac{2^{i-1}}{k}|V^{\act}_{i-1}|\]
	and hence the corresponding deterministic set $S_{i,j}$ is indeed a good set.
	To find such a bit $b_i$, it suffices to compute two things:
    \begin{itemize} 
    \item $\E[X|\text{first $i$ bits are fixed to $b_0,\ldots,b_{i-1},0$}]$, and 
    \item $\E[X|\text{first $i$ bits are fixed to $b_0,\ldots,b_{i-1}$}]$. 
    \end{itemize}
	It is possible to decompose $X$ into $X = \sum_{u \in V} X_u$ such that each node $u$, when given $b_0,b_1,\ldots,b_i$, $\alive_{i-1}(u)$ and $\phi_{i-1,j}(u)$, can efficiently compute $\E[X_u|b_0,b_1,\ldots,b_i]$, without any further communication. This in turn allows us to compute $\E[X|b_0,b_1,\ldots,b_i]$ in $O(D)$ rounds.
	Hence, given that every node knows $\alive_{i-1}(u)$ and $\phi_{i-1,j}(u)$, one can find a good set in $\tO(D \log n)$ \congest rounds, where $D$ denotes the diameter of the network. Hence, computing $\alive_{i-1}(u)$ and $\phi_{i-1,j}(u)$ can be done in $\tO(\separation \log n)$ rounds at the beginning of phase $i$. Moreover, $\alive_{i-1}(u)$ and $\phi_{i-1,j}(u)$ can be computed at the beginning of phase $i$ in $\tO(\separation \log n)$ rounds according to \cref{lem:low_degree_delays_communication}. Hence, the overall resulting run-time of this variant of \cref{alg:low_degree_delay_alg} is $\tO((D + \separation) \log^2(n)) = \tO(D \log^2 n)$. This is the complexity for the setting where we have a low-diameter global tree of depth $D$. One can replace this by a standard application of network decomposition to reduce the round complexity to $\poly(\log n)$. In particular, given a $c$-color $d$-diameter network decomposition of $G^{O(s)}$, we can use independent randomness for the nodes of different colors, and for each color, we can perform the gathering and bit fixing in $(s+d)\poly(\log n)$ rounds. Hence, we can perform the same derandomization in $(s+d)\poly(\log n)$ rounds. With the algorithm of \cite{elkin2022deterministic} that computes a $O(\log n)$-color $\poly(\log n)$-strong-diameter network decomposition in $s \poly(\log n)$ rounds~\cite{elkin2022deterministic}, this becomes a complexity of $s \poly(\log n)$ rounds overall for the whole derandomization procedure. Please see the proof of \Cref{thm:hittingset-pram} where we perform such a global derandomization via network decomposition for the hitting set problem and provide more of the lower-order details. Instead of diving into those details here, in this section, we focus on the local derandomization which leads to a faster round complexity, as discussed in the next subsection.
	
	\subsection{Algorithm $\mathcal{A}_{i,j}$ via Local Derandomization}
	\label{sec:low_degree_local_derandomization}

	This subsection is dedicated to providing the description of $\mathcal{A}_{i,j}$, that is proving \Cref{thm:deterministic-goodSet-selection} stated below. We note that this is the final missing piece in the proof of \cref{thm:low_degree_main}. 
	
	\begin{theorem}\label{thm:deterministic-goodSet-selection}
	For every iteration $j$ of phase $i$, there exists a \congest algorithm $\mathcal{A}_{i,j}$ which computes a good set $S_{i,j} \subseteq V^{\act}_{i-1}$ in $\tO(\separation \log n)$ rounds.
	\end{theorem}

 
	The algorithm $\mathcal{A}_{i,j}$ makes use of the local rounding framework of Faour et al.~\cite{Faour2022} to compute a good set $S_{i,j}$. Their rounding framework works via computing a particular weighted defected coloring of the vertices, which allows the vertices of the same color to round their values simultaneously, with a limited loss in some objective functions that can be written as summation of functions each of which depend on only two nearby nodes. Next, we provide a related definition and then state their black-box local rounding lemma.

    \begin{definition}[long-range d2-Multigraph]\label{def:long-range-d2multigraph}
  A long-range \emph{d2-multigraph} is a multigraph $H=(V_H,E_H)$ that is simulated on top of an underlying communication graph $G=(V,E)$ by a distributed message-passing algorithm on $G$. The nodes of $H$ are a subset of the nodes of $G$, i.e., $V_H\subseteq V$. The edge set $E_H$ consists of two kinds of edges, \emph{physical edges} and \emph{virtual edges}. Physical edges in $E_H$ are edges between direct neighbors in $G$. For each physical edge in $e\in E_H$ with $V(e)=\set{u,v}$, both nodes $u$ and $v$ know about $e$. Virtual edges in $E_H$ are edges between two nodes $u,v\in V_H$, and for each such virtual edge, there is a manager node $w$ which knows about this edge. 
  
  We next describe the assumed communication primitives. Let $M(v)$ be the set of nodes $w$ who manage virtual edges that include $v$. We assume $T$-round primitives that provide the following:   (1) each node $v$ can send one $O(\log n)$-bit message that is delivered to all nodes in $M(v)$ in $T$ rounds; (2) given $O(\log n)$-bit messages prepared at nodes $M(v)$ specific for node $v$, node $v$ can receive an aggregation of these messages, e.g., the summation of the values, in $T$ rounds.
\end{definition}
\begin{definition} (Pairwise Utility and Cost Functions) Let $H=(V_H,E_H)$ be a long-range d2-multigraph of an underlying communication graph $G=(V,E)$. For any label assignment $\vec{x}: V_H \rightarrow \Labels$, a pairwise utility function is defined as $\sum_{u \ in V_H} \utility(u, \vec{x}) + \sum_{e \in E_H} \utility(e, \vec{x})$, where for a vertex $u$, the function $\utility(u, \vec{x})$ is an arbitrary function that depends only on the label of $u$, and for each edge $e=\{u, v\}$, the function $\utility(e, \vec{x})$ is an arbitrary function that depends only on the labels of $v$ and $u$. These functions can be different for different vertices $u$ and also for different edges $e$. A pairwise cost function is defined similarly. For a probabilistic/fractional assignment of labels to vertices $V_{H}$, where vertex $v$ assumes each label in $\Sigma$ with a given probability, the utility and costs are defined as the expected values of the utility and cost functions, if we randomly draw integral labels for the vertices from their corresponding distributions (and independently, though of course each term in the summation depends only on the labels of two vertices and thus pairwise independence suffices).
\end{definition}

\begin{lemma}\label{lemma:long-range-d2rounding}[Faour et al.~\cite{Faour2022}]
  Let $H=(V_H,E_H)$ be a long-range d2-multigraph of an underlying communication graph $G=(V,E)$ of maximum degree $\Delta$, where the communication primitives have round complexity $T$. Assume that $H$ is equipped with pairwise utility and cost functions $\utility(\cdot)$ and $\cost(\cdot)$ (with label set $\Labels$) and with a fractional label assignment $\lambda$. Further assume that the given rounding instance is polynomially bounded in a parameter $q \leq n$. Then for every constant $c>0$ and every $\eps,\mu>\max\set{1/q^c, 2^{-c\sqrt{\log n}}}$, if $\utility(\lambda)-\cost(\lambda)>\mu\utility(\lambda)$, there is a deterministic \congest algorithm on $G$ to compute an integral label assignment $\ell$ for which $\utility(\ell)-\cost(\ell)\geq (1-\eps)\cdot\big(\utility(\lambda)-\cost(\lambda)\big)$ and such that the round complexity of the algorithm is
  \[
    T \cdot O\left(\frac{\log^2 q}{\eps\cdot\mu}\cdot\left(\frac{|\Labels| \log(q\Delta)}{\log n} + \log\log q \right)+\log q\cdot\log^* n \right).
  \]
\end{lemma}

    \paragraph{Our Local Derandomization.} In the following, for each node $u \in V$, we define $c_u = \frac{\phi_{i,j-1}(u)}{1- (|\alive_{i-1} (u)|/(10k))}$.
    The labeling space is whether each node in $V^{\act}_{i-1}$ is contained in $S_{i,j}$ or not, i.e., each node in $V^{\act}_{i-1}$ takes simply one of two possible labels $\Labels=\{0,1\}$ where $1$ indicates that the node is in $S_{i,j}$. For a given label assignment $\vec{x} \in \{0,1\}^{V^{\act}_{i-1}}$, we define the utility function
    
    \[\utility(\vec{x}) = \sum_{u \in V} c_u \sum_{v \in \alive_{i-1}(u)} x_v + \frac{2^{i-1}}{k}|V_{i-1}^\act| = \sum_{v \in V} \left( \sum_{u \in M_i(v)} c_u \right) x_v + \frac{2^{i-1}}{k}|V_{i-1}^{\act}|,\]
    and the cost function   
    \[\cost(\vec{x}) = \sum_{u \in V} c_u \sum_{v \neq v' \in \alive_{i-1}(u)} x_v x_{v'} + \sum_{v \in V^{\act}_{i-1}} 2^i x_v.\]

    If the label assignment is relaxed to be a fractional assignment $\vec{x} \in [0,1]^{V^{\act}_{i-1}}$, where intuitively now $x_v$ is the probability of $v$ being contained in $S_{i,j}$, the same definitions apply for the utility and cost of this fractional assignment.

    Note that the utility function is simply a summation of functions, each of which depends on the label of one vertex. Hence, it directly fits the rounding framework.  
    To capture the cost function as a summation of costs over edges, we next define an auxiliary multi-graph $H$ as follows: For each node $u \in V$ and every $v \neq v' \in \alive_i(u)$, we add an auxiliary edge between $v$ and $v'$, with a cost function which is equal to $c_u$ when both $v$ and $v'$ are marked, and zero otherwise.  
    Note that $H$ is a long-range d2-Multigraph where the communication primitives have round complexity $\tO(\separation \log n)$ as provided by \cref{lem:low_degree_delays_communication}.  
    
    We next argue that the natural fractional assignment where $x_v = \frac{1}{4k}$ for each $v\in V^{\act}_{i-1}$ satisfies the conditions of \Cref{lemma:long-range-d2rounding}. First, note that these fractional assignments are clearly polynomially bounded in $q$ for $q=k=O(\log\log n)$. Next, we discuss that, for the given fractional assignment, utility minus cost is at least a constant factor of utility.
    
    \begin{claim}
    Let $\vec{x} \in [0,1]^{V^{\act}_{i-1}}$ with $x_v = \frac{1}{4k}$ for every $v \in V^{\act}_{i-1}$. Then, $\utility(\vec{x}) - \cost(\vec{x}) \geq \utility(\vec{x})/2$.
    \end{claim}
    \begin{proof}
    We have
    \begin{align*}
        \utility(\vec{x}) &= \sum_{u \in V} c_u \sum_{v\in \alive_{i-1}(u)} x_v + \frac{2^{i-1}}{k}|V_{i-1}^{\act}|  \\
        &\geq 2 \left(\sum_{u \in V}  c_u \sum_{v,v' \in \alive_{i-1}(u)} x_v x_{v'} + \sum_{v \in V^{\act}_{i-1}} 2^i \frac{1}{4k} \right) \\
        &\geq 2 \cost(\vec{x}).
    \end{align*}
    and therefore indeed $\utility(\vec{x}) - \cost(\vec{x}) \geq \utility(\vec{x})/2$.
    \end{proof}
    
    Hence, we can apply \Cref{lemma:long-range-d2rounding} on these fractional assignments with $\mu=1/2$ and $\eps=0.1$, which runs in $\tO((\log\log \log n)^2)$ iterations of calling the communication primitives, each taking $\tO(s\log^2 n)$ rounds. Hence, the entire procedure runs in $\tO(s\log^2 n)$ rounds. As a result of applying 
    \Cref{lemma:long-range-d2rounding} with these parameters, we get an integral label assignment 
    $\vec{y} \in \{0,1\}^{V^\act_{i-1}}$ which satisfies $\utility(\vec{y}) - \cost(\vec{y}) \geq 0.9 (\utility(\vec{x}) - \cost(\vec{x}))$. We can then conclude
    \begin{align*}
        \utility(\vec{y}) - \cost(\vec{y}) &\geq 0.9 (\utility(\vec{x}) - \cost(\vec{x})) \\
        &\geq0.9 \left( \sum_{u \in V} c_u \frac{|\alive_{i-1}(u)|}{4k} + \frac{2^{i-1}}{k}|V^{\act}_{i-1}| - \left(    \sum_{u \in V} c_u \frac{|\alive_{i-1}(u)|}{16 k} + \frac{2^{i-2}}{k}|V^{\act}_{i-1}|\right) \right)\\
        &\geq \sum_{u \in V} c_u \frac{|\alive_{i-1}(u)|}{10k}.
    \end{align*}
    This integral label assignment directly gives us $S_{i,j}$. In particular, let $S_{i,j} = \{v \in V^{\act}_{i-1} \colon y_v = 1\}$. Note that
    \begin{align*}
        \utility(\vec{y}) - \cost(\vec{y}) = \sum_{u \in V} c_u \left( |\alive_{i-1}(u) \cap S_{i,j}| - \binom{|\alive_{i-1}(u) \cap S_{i,j}|}{2} \right) + \frac{2^{i-1}}{k}|V^{\act}_{i-1}| - 2^i|S_{i,j}|, 
    \end{align*}
    and therefore
    \begin{align*}
        \sum_{u \in V} Y_{i,j}(u)\frac{\phi_{i,j-1}(u)}{1 - (|\alive_{i-1}(u)|/(10k))} + |S_{i,j}| \cdot 2^i 
        &= \sum_{u \in V} c_u - \utility(\vec{y}) + \cost(\vec{y}) + \frac{2^{i-1}}{k}|V^{\act}_{i-1}| \\
        &\leq \sum_{u \in V} c_u - \sum_{u \in V} c_u \frac{|\alive_{i-1}(u)|}{10k}  + \frac{2^{i-1}}{k}|V^\act_{i-1}| \\
        &\leq \sum_{u \in V} \phi_{i,j-1}(u) + \frac{2^{i-1}}{k}|V^{\act}_{i-1}|,    \end{align*}
    which shows that $S_{i, j}$ is indeed a good set according to \Cref{def:low_degree_del_good_set}. This completes the description of our locally derandomized construction of good sets $S_{i, j}$, hence completing the proof of \Cref{thm:deterministic-goodSet-selection}.

	\section{From Low-Degree Clusters to Isolated Clusters}
	\label{sec:LowDegtoIsolation}
    \subsamplingmain
    
    Similar as in \cref{sec:low_degree}, we could get the same guarantees with a \congest algorithm with round complexity $O(s \poly(\log n))$ by performing a global derandomization with the help of a previously computed network decomposition.
    
    \begin{proof}[Proof of \Cref{thm:subsampling_main}]
    The clustering $\fC^{out}$ is computed in two steps.
    In the first step, we use the local rounding procedure to compute a clustering $\fC'$ which one obtains from $\fC$ by only keeping some of the clusters in $\fC$ (any such cluster is kept in its entirety). Intuitively, the local rounding procedure derandomizes the random process which would include each cluster $C$ from $\fC$ in the clustering $\fC'$ with probability $\frac{1}{2k}$, $k = \lceil 100 \log \log n\rceil$, pairwise independently.
        Given the clustering $\fC'$, we keep each node $u \in \fC'$ clustered in $\fC^{out}$ if and only if the $\separation$-hop degree of $u$ in $\fC'$ is $1$.
		Note that given $\fC'$, the output clustering $\fC^{out}$ can be computed in $\tO(\separation \log n)$ \congest rounds.
	    
        First, we discuss the first property, i.e., the strong diameter of the output clustering. The fact that the clustering $\fC$ has strong diameter $O(\separation \log (n))$ directly implies that the clustering $\fC'$ also has strong diameter $O(\separation \log(n))$, simply because each cluster of $\fC'$ is exactly one of the clusters of $\fC$. We next argue that $\fC^{out}$ also has strong diameter $O(\separation \log n)$. Let $u$ be a node clustered in $\fC'$ and $P_u$ the unique path from $u$ to its center in the tree associated with its cluster. Then, it directly follows from the definition that for every $w \in P_u$, the $\separation$-hop degree of $w$ in $\fC'$ is at most the $\separation$-hop degree of $u$ in $\fC'$. Therefore $u$ being clustered in $\fC^{out}$ implies that $w$ is also clustered in $\fC^{out}$. Hence, we conclude that $\fC^{out}$ indeed has strong diameter $O(\separation \log n)$. 
		
		Next, we discuss the second property: The clustering $\fC^{out}$ is $\separation$-hop separated. This directly follows from the fact that by definition every clustered node has a $\separation$-hop degree of $1$.
		
		Finally, To prove \Cref{thm:subsampling_main}, the only remaining thing is to prove the third property, i.e., that we compute $\fC'$ in such a way that $\fC^{out}$ clusters at least $\frac{|\fC|}{1000 \log \log (n)}$ nodes. The rest of this proof is dedicated to this property.
    
    For each cluster $C \in \fC$, we let $center(C)$ denote the cluster center of $C$ and define $Centers = \{center(C) \colon C \in \fC\}$ as the set of cluster centers of $\fC$.
    Moreover, for each $u$ clustered in $\fC$, recall that $c_u$ is the cluster center of the cluster of $u$ and let $P_u$ be the unique $u$-$c_u$ path in the tree associated with this cluster $C_u$. Now, let
    \[S_u = \{C \in \fC \colon d(P_u,C) \leq \separation\}.\]
    Note that the size of $S_u$ is equal to the $\separation$-hop degree of $u$, which by assumption is at most $k$.

    The labeling space is for each cluster center whether its cluster is contained in $\fC'$ or not, i.e., each node in $Centers$ takes simply one of two possible labels $\{0,1\}$ where $1$ indicates that the corresponding cluster is in $\fC'$. For a given label assignment $\vec{x} \in \{0,1\}^{Centers}$, we define 
    \[\utility(\vec{x}) = \sum_{C \in \fC} |C|x_{center(C)}\]
    and    
    
    \[\cost(\vec{x}) = \sum_{u \in V \colon \text{$u$ is clustered in $\fC$}}\sum_{C \in S_u \setminus {C_u}} x_{c_u}x_{center(C)}.\]

    If the label assignment is relaxed to be a fractional assignment $\vec{x} \in [0,1]^{Centers}$, where intuitively now $x_v$ is the probability of $v$'s cluster being contained in $\fC'$, the same definitions apply for the utility and cost of this fractional assignment.
    The utility function is simply a summation of functions, each of which depends on the label of one vertex in $Centers$. Hence, it directly fits the rounding framework.
    
    To capture the cost function as a summation of costs over edges, we next define an auxiliary multi-graph $H$ as follows: For each node $u $ clustered in $\fC$ and every $C_1 \neq C_2 \in S_u$, we add an auxiliary edge between $center(C_1)$ and $center(C_2)$ with a cost function which is equal to $1$ when both $C_1$ and $C_2$ are contained in $\fC'$, and zero otherwise.  
    Note that $H$ is a long-range d2-Multigraph, according to \Cref{def:long-range-d2multigraph}. The communication primitives can be implemented in $\tO(\separation \log n)$ rounds according to the lemma below.

    \begin{lemma}
      \label{lem:subsampling_delays_communication}
      Let $\fC$ be the input clustering of \cref{thm:subsampling_main}. There exists a \congest algorithm running in $\tO(\separation \log n)$ rounds which computes for each node $u \in V$ the sets $\{center(C) \colon C \in S_u\}$. Moreover, for each $v \in Centers$, let $M(v) = \{u \in V \colon C_v \in S_u\}.$ Then, there exists an $\tO(\separation \log n)$ round \congest algorithm that allows each node $v$ to send one $O(\log n)$-bit message that is delivered to all nodes in $M(v)$. Similarly, there also exists an $\tO(\separation \log n)$ round \congest algorithm that given $O(\log n)$-bit messages prepared at nodes in $M(v)$ specific for node $v$, it allows node $v$ to receive an aggregation of these messages, e.g., the summation of the values, in $\tO(\separation \log n)$ rounds. 
    \end{lemma}
    \begin{proof}[Proof of \Cref{lem:subsampling_delays_communication}] 
    The proof follows roughly along the lines of the proof of \cref{lem:low_degree_delays_communication}. First, we run the following variant of breadth first search: At the beginning, each node clustered in $\fC$ has a token which is equal to the identifier of its cluster.
    Now, in each of the $\separation$ iterations, each node that has received at most $k = \lceil 100 \log \log n \rceil$ identifiers in the previous iteration forwards all the identifiers it has received to its neighbors. If a node has received more than $k$ identifiers, it selects $k$ of them to forward.
    The first phase can be implemented in $O(k \separation)$ \congest rounds.
    It directly follows from the fact that the $\separation$-hop degree of $\fC$ is at most $k$ that after the first phase each node $w$ learns the identifiers of all cluster centers such that the corresponding cluster $C$ satisfies $d(w,C) \leq \separation$. The next phase propagates this information up in the cluster tree, from the root toward the leaves, such that each descendant of $w$---i.e., any node whose cluster path to the root passes through $w$---learns about all those cluster centers as well.
    The second phase consists of $O(\separation \log n)$ iterations. In each iteration, each
    clustered node sends all the identifiers it learned about so far to each of its children in the corresponding cluster tree.
    It again follows from the fact that the $\separation$-hop degree of $\fC$ is at most $k$ that each of the $O(\separation \log n)$ iterations in the second phase can be implemented in $O(k)$ \congest rounds. Hence, the overall \congest runtime is $\tO(\separation \log n)$.

    For each $v \in Centers$, let $M(v) = \{u \in V \colon C_v \in S_u\}.$ By repeating the above communication, we have a $\tO(\separation \log n)$-round procedure that delivers one message from each node $v$ to all nodes $M(v)$. By reversing the same communication in time, we can also provide the opposite direction: if each node in $M(v)$ starts with a message for $v$, then in $\tO(\separation \log n)$ rounds, we can aggregate these messages and deliver the aggregate to $v$, simultaneously for all $v$.
    \end{proof}
    
    \begin{claim}
    Let $k= \lceil 100 \log \log (n) \rceil$ and $\vec{x} \in [0,1]^{Centers}$ with $x_v = \frac{1}{2k}$ for every $v \in Centers$. Note that this fractional label assignment is polynomially bounded in $q=k=O(\log\log n)$. Furthermore, we have $\utility(\vec{x}) - \cost(\vec{x}) \geq \utility(\vec{x})/2$.
    \end{claim}
    \begin{proof}
    We have
    \[\utility(\vec{x}) = \sum_{C \in \fC} |C|x_{center(C)} = \sum_{C \in \fC} \frac{|C|}{2k} = \frac{|\fC|}{2k},\]
    and 
       \[\cost(\vec{x}) = \sum_{u \in V \colon \text{$u$ is clustered in $\fC$}}\sum_{C \in S_u \setminus \{C_u\}} x_{c_u}x_{center(C)} \leq \sum_{u \in V \colon \text{$u$ is clustered in $\fC$}} \frac{1}{4k} \frac{|S_u|}{k} \leq \frac{|\fC|}{4k}.\]
    Therefore, indeed $\utility(\vec{x}) - \cost(\vec{x}) \geq \utility(\vec{x})/2$.
    \end{proof}
    
    We now invoking the rounding of \Cref{lemma:long-range-d2rounding} with parameters $\mu=0.5$, $\eps=0.5$, and $q=k=O(\log\log n)$ on the fractional label assignment of $\vec{x} \in [0,1]^{Centers}$ where $x_v = \frac{1}{2k}$ for every $v \in Centers$. The procedure runs in $\tO(s\log n)$ rounds. As output, we get an integral label assignment $\vec{y} \in \{0,1\}^{Centers}$ which satisfies
    \[\utility(\vec{y}) - \cost(\vec{y}) \geq  0.5 (\utility(\vec{x}) - \cost(\vec{x})) \geq  \frac{|\fC|}{8k}.\]

    Let $C' = \{C \in \fC \colon y_{center(C)} = 1\}$. Note that for every $u \in \fC$,
    \[I(\text{$u$ is clustered in $\fC^{out}$}) \geq y_{c_u} - \sum_{C \in S_u \setminus \{C_u\}} y_{c_u} y_{center(C)}.\]
    Therefore,
    \begin{align*}
        |\fC^{out}|&\geq \sum_{u \in V \colon \text{$u$ is clustered in $\fC$}} I(\text{$u$ is clustered in $\fC^{out}$}) \\
        &\geq \sum_{u \in V \colon \text{$u$ is clustered in $\fC$}} \left(y_{c_u} - \sum_{C \in S_u \setminus \{C_u\}} y_{c_u} y_{center(C)}\right)\\
        &= \utility(\vec{y}) - \cost(\vec{y}) \\
        &\geq \frac{|\fC|}{8k}
    \end{align*}
    and therefore $\fC^{out}$ clusters enough vertices to prove \Cref{thm:subsampling_main}.

    \end{proof}
	
	\section{Clustering More Nodes}
	\label{sec:clusteringmorenodes}
	In this section, we prove the following result, which says that once we have access to a clustering algorithm that clusters a nontrivial proportion of nodes with sufficient separation, we can turn it into an algorithm that clusters a constant proportion of nodes. We are paying for this with a slight decrease in the separation guarantees.

\clusteringmorenodesmain

It follows from the analysis of \cref{alg:clustering_with_more_vertices} and its subroutine \cref{alg:expanding}. 
To understand the pseudocode of the algorithms, we note that for a set of nodes $C \subseteq V(G)$ and $D \in \mathbb{N}_0$, we define $$C^{\leq D} = \{v \in V \colon d(C,v) \leq D\}.$$
Moreover, we say that a cluster $C$ is \emph{good} in \cref{alg:expanding} if $\cut(C) < +\infty$. Otherwise, $C$ is \emph{bad}.

	\begin{algorithm}
		\caption{Making a clustering algorithm cluster half of the nodes}
		\label{alg:clustering_with_more_vertices}
		\begin{algorithmic}[1] 
			\Procedure{\textsc{ClusterHalfNodes($G$)}}{} 
			\State $\fC_0 = \emptyset$
			\State $N = \lceil 4 \cdot 2^x\rceil$
			\For{$i = 1,2 \ldots, N$}
			\State $G_i = G\left[V \setminus \left( \bigcup_{C \in \fC_{i-1}} C \right)^{\le 1} \right]$
			\State $\fC \leftarrow \fA(G_i)$
			\State $\hfC_i \leftarrow \textsc{Expand}(G_i, \fC)$
			\State $\fC_i = \fC_{i-1} \cup \hat{\fC}_i$
			\EndFor
			\Return $\fC_N$
			\EndProcedure
		\end{algorithmic}
	\end{algorithm}

    \begin{algorithm}
		\caption{Expanding an input clustering}
		\label{alg:expanding}
		{\bf Input:} A graph $G$ and its $10x$-separated clustering $\fC$\\
	{\bf Output:} An expanded clustering $\hfC$ with small boundary
	\begin{algorithmic}[1] 
			\Procedure{\textsc{Expand($G, \fC$)}}{} 
			\For{$C \in \fC$}
                \State Define $\cut(C) = \min \left\{0 \le i \le 3x \colon |C^{\leq i+1}| \leq 1.5 |C^{\leq i}| \right\}$ and $\cut(C) = +\infty$ if no such $i$ exists.
                \State If $\cut(C) < +\infty$, define $\expand(C)  = C ^{\textrm{cut}(C)}$
            \EndFor
			\State \Return $\hfC = \{\expand(C) : C \in \fC, \cut(C) < +\infty\}$
			\EndProcedure
		\end{algorithmic}
	\end{algorithm}

    We start by analyzing \cref{alg:expanding} in the following lemma. 	Importantly, the fourth condition for $\hfC$ in the statement below states that the total number of unclustered vertices neighboring one of the clusters in $\hat{\fC}$ is at most half the total number of clustered vertices. This is the reason why we can 

	\begin{lemma}
		
		\label{lem:clustering_transformed}
		Let $x \geq 2$ be arbitrary and $\fC$ a clustering with
		
		\begin{enumerate}
			\item strong diameter $O(x \log n)$,
			\item separation $10 x$ and
			\item clustering at least $\frac{n}{2^x}$ nodes.
		\end{enumerate}

		Then, $\hat{\fC}$ constructed in \cref{alg:expanding} is a clustering with 
		
		\begin{enumerate}
			\item strong diameter $O(x \log n)$,
			\item separation $4 x$,
			\item clustering at least $0.5\frac{n}{2^x}$ nodes and
			\item $\left|\left(\bigcup_{C \in \hfC} C\right)^{\leq 1}\right| \leq 1.5 \left|\left(\bigcup_{C \in \hfC} C\right)^{\leq 1}\right|$.
		\end{enumerate}
		
		Moreover, the algorithm can be implemented in $O(x \log n)$ \congest rounds. 
	\end{lemma}

	\begin{proof}
		The first property follows from the fact that for a set $S$ and $D \in \mathbb{N}_0$, $\diam(S^{\leq D}) \leq \diam(S) + 2D$. Hence, for a good cluster $C$,
		
		\[\diam(\expand(C)) \leq \diam(C) + 2\cut(C) = O(x \log n).\]
		
		To prove the second property, let $C_1 \neq C_2 \in \fC$ be two arbitrary good clusters. For $i \in \{1,2\}$, let $u_i \in \expand(C_i)$ be arbitrary. By triangle inequality, we have: 
		
		\[d(u_1,u_2) \geq d(C_1,C_2) - d(C_1,u_1) - d(C_2,u_2) \geq 10  x - 2 \cdot 3x \geq 4x.\]
		
		To prove the third property, it suffices to show that at most $0.5 \frac{n}{2^x}$ of the nodes are contained in bad clusters. For a bad cluster $C$, a simple induction implies $|C^{\leq 3x}| \geq 1.5^{3x} |C| \geq 2 \cdot 2^x |C|$.  Therefore,
		
		\[\sum_{C \in \fC, \text{$C$ is a bad cluster}} |C| \leq \frac{1}{2^{x+1}} \sum_{C \in \fC, \text{$C$ is a bad cluster}} |C^{\leq 3x}| \leq \frac{n}{2^{x+1}}, \]
		where the last inequality follows from the fact that for two clusters $C_1 \neq C_2 \in \fC$, $C_1^{\leq 3x} \cap C_2^{\leq 3x} = \emptyset$.
		
		To prove the fourth property we write
		\begin{align*}
				\left|\left(\bigcup_{C \in \hat{\fC}} C\right)^{\leq 1}\right| &\leq \sum_{\hat{C} \in \hat{\fC}} |\hat{C}^{\leq 1}| \\
				&= \sum_{C \in \fC \colon \text{$C$ is a good cluster}} |(C^{\leq \cut(C)})^{\leq 1}| \\
				&= \sum_{C \in \fC \colon \text{$C$ is a good cluster}} |C^{\leq \cut(C) + 1}|  \\
				&\leq 1.5 \sum_{C \in \fC \colon \text{$C$ is a good cluster}} |C^{\leq \cut(C)}| \\
				&= 1.5 \left|\left(\bigcup_{C \in \hfC} C\right)^{\leq 1}\right|.
		\end{align*}
	
		It remains to discuss the \congest computation. Since we have for any $C_1, C_2 \in \fC$ that $C_1^{\le 3x} \cap C_2^{\le 3x} = \emptyset$, each cluster $C \in \fC$ can compute the values of $C_1^{\le 0}, C_1^{\le 1}, \dots, C_1^{\le 3x}$ by running one breadth first search from $C_1$ up to distance of $3x$.  
		
	\end{proof}

We are now ready to prove \cref{thm:clusteringmorenodesmain}. 

\begin{proof}[Proof of \cref{thm:clusteringmorenodesmain}]
We show that the algorithm satisfies the following invariants for $i \in \{0,1,\ldots,N\}$:
	
	\begin{enumerate}
		\item $\fC_i$ is $2$-separated
		\item $|V(\fC_i)| \geq n \cdot \min(0.5,\frac{i}{8 \cdot 2^x})$
		\item $|V(\fC_i^{\leq 1})| \leq 1.5 |V(\fC_i)| $
	\end{enumerate}

	The base case $i = 0$ trivially holds. Now, consider an arbitrary $i \in [N]$ and assume that the invariant is satisfied for $i-1$. To check the first invariant, let $C_1 \neq C_2 \in \fC_i$ be arbitrary. If $C_1,C_2 \in \fC_{i-1}$, then it follows by induction that $d(C_1,C_2) \geq 2$. If $C_1,C_2 \in \hat{\fC}_i$, then it follows from \cref{lem:clustering_transformed} that $d_{G_i}(C_1,C_2) \geq 2$ which also directly implies $d_G(C_1,C_2) \geq 2$. It remains to consider the case that one cluster, let's say $C_1$, is in $\fC_{i-1}$ and $C_2$ is in $\hat{\fC}_i$.
	We have 
	
	\[C_2 \subseteq V(G_i) = V \setminus V(\fC_{i-1}^{\leq 1}) \subseteq V \setminus C_1^{\leq 1}\]
	
	and therefore $d(C_1,C_2) \geq 2$, as desired. 
	
	Next, we show that the second invariant is preserved. If $|V(\fC_{i-1})| \geq n/2$, then there is nothing to show. Otherwise, we have
	
	\[|V(G_i)| \geq n - |V(\fC_{i-1}^{\leq 1})| \geq n - 1.5 |V(\fC_{i-1})| \geq n - 1.5\frac{n}{2} = \frac{n}{4}.\]
	
	Therefore, according to \cref{lem:clustering_transformed}, $\hat{C}_i$ clusters at least $0.5 \frac{(n/4)}{2^x} = \frac{n}{8 \cdot 2^x}$ vertices, which together with $|V(\fC_{i-1})| \geq n \cdot \min(0.5,\frac{i-1}{8  \cdot 2^x})$ directly implies $|V(\fC_i)| \geq n \cdot \min(0.5,\frac{i}{8  \cdot 2^x})$.
	It remains to verify the third property. According to \cref{lem:clustering_transformed}, we have
	
	\[|V(\hat{\fC}^{\leq 1}_i)  \setminus V(\fC^{\leq 1}_{i-1})| \leq 1.5 |V(\hat{\fC}_i)|.\]
	
	Therefore,
	\[|V(\fC^{\leq 1}_i)| = |V(C_{i-1}^{\leq 1})| + |V(\hat{\fC}^{\leq 1}_i)  \setminus V(\fC^{\leq 1}_{i-1})| \leq 1.5 |V(\fC_{i-1})| + 1.5 |V(\hat{\fC}_i)| = 1.5 |V(\fC_i)|.\]
	
	This finishes the proof that the invariants are satisfied throughout the algorithm. Hence, $\fC_N$ is a $2$-separated clustering that clusters at least half of the vertices. Moreover, it directly follows from the strong diameter guarantee of \cref{lem:clustering_transformed} that $\fC_N$ has strong diameter $O(x \log n)$. Finally, as $\hat{\mathcal{A}}$ has a round complexity of $O(x \log n)$, it follows that $\fC_N$ is computed in $O(2^x(R + x \log n))$ \congest rounds. This concludes the proof of \cref{lem:pairwise_del_clustering}.
\end{proof}

\section{Hitting Set}
\label{sec:hittingset}
In this section, first, we introduce a variant of the hitting set problem. Next, we propose a simple randomized algorithm for this problem using only pairwise independence. In the end, we describe an efficient distributed/parallel derandomization of our randomized algorithm.

\subsection{Problem Definition}
Consider a collection $\mathcal{S} = \{S_1, \dots, S_N\}$ of $N$ subsets from the universe $\{1,\dots,n\}$ and let $w_i \geq 0$ be the weight that is assigned to $S_i$. We say a subset $H \subseteq [n]$ hits $S_i$ if $H \cap S_i \neq \emptyset$. Our goal is to find a small $H$ with a small cost. Cost of $H$ is total weights of $S_i$ that are not hit by $H$, i.e., $\sum_{i:S_i \cap H = \emptyset} w_i$. For a random subset $H$ that includes each element with probability $p$ independently, the expected size of $H$ is $\E[|H|] = np$ and its expected cost is
\begin{equation*}
    \sum_{i=1}^{N} w_i (1-p)^{|S_i|} \approx \sum_{i=1}^{N} w_i e^{-|S_i|p} = \tau^p_{\mathcal{S}}.
\end{equation*}
For example, suppose the regular case where $|S_i| = \Delta$. For $p = 10\log N / \Delta$, a random subset hits all sets with high probability $1 - 1/\poly(N)$ and for $p = 1 / \Delta$, constant fraction of sets are hit. Two important examples for weights is when $w_i = 1$ and $w_i = |S_i|$. In the former, we simply count the number of not hit sets. The latter indeed appears in our applications for constructing spanners and distance oracles (see  \Cref{sec:hittingset-app}). There, we get penalized for each not hit set by its size.

In many applications, the expected size and cost of a random subset are enough. The challenge is to find a subset deterministically. Based on this, we formulate the following problem where we combine our two objectives in one potential function.

\label{subsec:hittingset-define}
\begin{definition}[Hitting Set Problem]
\label{def:hittingset}
Given a collection $\mathcal{S} = \{S_1, \dots, S_N\}$ of $N$ subsets from the universe $\{1,\dots,n\}$, an integer weight $w_i \geq 0$ for each $S_i$, and a sampling parameter $p \in (0,1)$, find a subset $H$ that minimizes the potential function
\begin{equation}
\label{eq:hittingset-potential}
\Phi^{p}_{\mathcal{S}}(H) = \frac{\sum_{i=1}^{N} w_i \cdot \mathbb{1}[H\cap S_i = \emptyset]}{\tau^p_{\mathcal{S}}} + \frac{|H|}{np}.
\end{equation}
\end{definition}

So if $\Phi^{p}_{\mathcal{S}}(H) = O(1)$, then $H$ has size $O(np)$ and its cost is $O(\tau^{p}_{\mathcal{S}})$. Our goal is to find such a set with constant potential function deterministically and efficiently. In the rest, we assume that $N \geq n$ as we can add dummy sets with zero costs. We also assume that $p \leq 1/2$ to ensure that $1 - p = e^{-\Theta(p)}$. Note that the case $p \geq 1/2$ is trivial since we tolerate constant deviation from a random subset and for $p \geq 1/2$, the expected size of a random subset is at least $n/2$. So our hitting can include all the $n$ elements.

\paragraph{Hitting in Ordered Sets.}
There are applications where $H$ is partially penalized even if we hit $S_i$. The amount of cost depends on which element of $S_i$ is being hit. In \Cref{sec:hittingset-app}, we encounter a particular instance of this generalization which is described in the following.

For each $S_i$, there is no weight but there is an order $\pi_i(\cdot)$ on its elements where $\pi_i(j)$ denotes the $j$-th element of $S_i$ for $j=1,\dots,|S_i|$. Then, $H$ has to pay $k-1$ for $S_i$ if $\pi_i(k) \in H$ and 
\begin{equation*}
    H \cap \{\pi_i(1), \dots, \pi_i(k-1)\} = \emptyset.
\end{equation*}
If $H$ does not hit $S_i$ at all, it has to pay $|S_i|$. Cost of $H$ is the sum of the expenses incurred by each $S_i$. We call this problem \textit{hitting ordered set}. With this definition, the expected cost for a random $H$ is
\begin{equation*}
    \sum_{i=1}^{N} \sum_{j=1}^{|S_i|} (1-p)^j
\end{equation*}
The hitting ordered set problem is related to the original setting of \Cref{def:hittingset} in the following sense.

\begin{lemma}
\label{lem:hittingset-ordered-reduction}
Given an instance $\mathcal{I}_1$ of the hitting ordered set problem with $N$ sets $S_1, \dots, S_N \subseteq [n]$, we can construct an instance $\mathcal{I}_2$ of the original hitting set problem (see \Cref{def:hittingset}) with $O(N \log n)$ sets in $O(\sum_{i=1}^{N} |S_i|)$ time such that the following holds: For any $H \subseteq [n]$, if $c_1$ is the cost of $H$ in $\mathcal{I}_1$ and $c_2$ is the cost of $H$ in $\mathcal{I}_2$, then $c_1 \leq c_2 \leq 3c_1$.
\end{lemma}

\begin{proof}
To construct $\mathcal{I}_2$, for each $S_i$ in $\mathcal{I}_1$, we add $O(\log n)$ sets to $\mathcal{I}_2$. Suppose $2^\ell \leq |S_i| < 2^{\ell + 1}$. For $j\in [\ell]$, let $S_i^{j} = \{\pi_i(1), \dots, \pi_i(2^j)\}$ and let $S_i^{\ell+1} = S_i$. This completes the construction of sets of $\mathcal{I}_2$. Weight of $S_i^j$ in $\mathcal{I}_2$ is its size $|S_i^j|$.

Consider a subset $H \subseteq [n]$ and let $k$ be the minimum index that $\pi_i(k) \in H$. Suppose $k$ is $|S_i|+1$ if there is no such index. So $H$ has to pay $k-1$ in $\mathcal{I}_1$. In $\mathcal{I}_2$, it has to pay $\sum_{j:|S_i^j| < k} |S_i^j|$ which lies in the range $[(k-1), 3(k-1)]$ and concludes the proof.
\end{proof}

\subsection{Iterative Sampling}
\label{subsec:hittingset-alg}
The goal of this section is to find $H$ with $\Phi^{p}_{\mathcal{S}} = O(1)$ for the hitting set problem \Cref{def:hittingset}. Let $\Delta = \max_{i \in [N]} |S_i|$. Our algorithm has $T = \lceil 8p\Delta\rceil$ iterations. We start with a randomized algorithm and then we derandomize it. For $t = 1,\dots, T$, let $\mathcal{P}^t$ be a pairwise-independent distribution over $n$ binary random variables $X^t_1, \dots, X^t_n \in \{0,1\}$ with bias $q = 4p/T$. That is:
\begin{align*}
    \forall i \in [n], \forall b\in \{0,1\},&\quad \Pr[X^t_i = b] = q^b (1-q)^{(1-b)},\\
    \forall i,j \in [n], i \neq j,\forall b,b' \in \{0,1\},&\quad \Pr[X^t_i = b, X^t_j = b'] = q^{b+b'} (1-q)^{2-(b+b')}.
\end{align*}
Let the random subset $G^t$ be $\{i \in [n] \mid X^t_i = 1\}$. We replace $G^t$s one by one with an explicit set $H^t$. The final output of the algorithm is $H = \cup_{t=1}^T H^t$. Suppose we are in iteration $t$. Our goal is to find $H^t$. Let
\begin{equation}
\label{eq:pessimistic-def}
    Y^t_i = \sum_{j \in S_i} X^t_j - \sum_{j \in S_i} \sum_{k \in S_i: j < k} X^t_j X^t_k.
\end{equation}
If $G^t$ does not hit $S_i$, then $Y^t_i = 0$. Otherwise, $Y^t_i \leq 1$ (because $a \leq \binom{a}{2}+1$ for all positive integers $a$). So $1 - Y^t_i$ is always greater than or equal to $\mathbb{1}[G^t \cap S_i = \emptyset]$ and is a pessimistic estimator for the event that $G^t$ does not hit $S_i$. We have the following upper bound on $\E[1 - Y^t_i]$.

\begin{lemma}
\label{lem:hittingset-pessimistic-bound}
$\E[1-Y^t_i] \leq 1 - 3|S_i|p/T \leq e^{-|S_i|p/T}.$
\end{lemma}

\begin{proof}
Note that:
\begin{equation*}
    \E[Y^t_i] = |S_i|q - \binom{|S_i|}{2}q^2 \geq |S_i|q - |S_i|^2q^2/2 \geq 3|S_i|q/4 = 3|S_i|p/T
\end{equation*}
where in the last inequality we use $q = 4p/T \leq 1/2\Delta \leq 1/2|S_i|$.
\end{proof}

For a subset $G \subseteq [n]$, we define the function $f^t(G)$ as
\begin{equation*}
    f^t(G) = \frac{\sum_{i: S_i \cap (H^1 \cup \dots \cup H^{t-1}) = \emptyset} (1 - Y_i) \cdot w_i e^{-|S_i|(T - t)p/T}}{\tau^p_{\mathcal{S}}} + \frac{\sum_{i=1}^{n} X_i + \sum_{j=1}^{t-1} |H^j| + 4n(T - t)p/T}{4np}
\end{equation*}
where $X_i = \mathbb{1}[i \in G]$ and $Y_i$ is defined from $X_1, \dots, X_n$ similar to \Cref{eq:pessimistic-def}.

\begin{lemma}
\label{lem:hittingset-first-iteration}
$\E[f^1(G^1)] \leq 2.$
\end{lemma}

\begin{proof}
Note that $\E[\sum_{i=1}^{n} X^1_i] = nq = 4np/T$ and from \Cref{lem:hittingset-pessimistic-bound}, we have $\E[1 - Y_i^1] \leq e^{-|S_i|p/T}$. Plugging these two bounds completes the proof.
\end{proof}

\begin{lemma}
\label{lemma:hittingset-invariant}
For $t \geq 2$, we have:
\begin{equation*}
    \E[f^t(G^t)] \leq f^{t-1}(H^{t-1}).
\end{equation*}
\end{lemma}

\begin{proof}
Consider a subset $S_i$. If one of $H^1, \dots, H^{t-2}$ hits $S_i$, then the contribution of $S_i$ to the both sides of the inequality is zero. Otherwise, if $H^{t-1}$ hits $S_i$, the contribution of $S_i$ to $\E[f^t(G^t)]$ is zero. Note that it may contribute a non-zero amount into the RHS since we use pessimistic estimator $1 - Y_i$. The only remaining case is when $S_i$ is not hit in any of the first $t-1$ iterations. Then, the contribution of $S_i$ to the LHS is \begin{equation*}
\E[1 - Y^t_i] \cdot w_i e^{-|S_i|(T - t)p/T} \leq w_i e^{-|S_i|(T - t + 1)p / T}
\end{equation*}
where we use \Cref{lem:hittingset-pessimistic-bound}. On the other hand, the contribution of $S_i$ to the RHS is exactly $w_i e^{-|S_i|(T - t + 1)p / T}$. So the contribution of each $S_i$ to the LHS is less than or equal to its contribution to the RHS. Since $\E[|G_t|] = nq = 4np/T$, the second term that controls the size in $f^t(\cdot)$ and $f^{t-1}(\cdot)$ are equal which completes the proof.
\end{proof}

\begin{theorem}
\label{thm:hittingset-sampling}
If $f^t(H^t) \leq \E[f^t(G^t)]$ for all $t = 1, \dots, T$, then 
\begin{equation*}
    \Phi^p_{\mathcal{S}}(H = H^1 \cup \dots \cup H^T) \leq 2.
\end{equation*}
\end{theorem}
\begin{proof}
From \Cref{lem:hittingset-first-iteration} and \Cref{lemma:hittingset-invariant}, we get that $f^T(H^T) \leq 2$. Comparing $f^T(H^T)$ and $\Phi^p_{\mathcal{S}}(H)$ term by term, we can easily see that $f^T(H^T) \geq \Phi^p_{\mathcal{S}}(H)$.
\end{proof}

If $\Delta \gg 1/p$, then the number of iterations can be quite large. However, we are mostly interested in the regime where the number of iterations is logarithmic. We can achieve this as stated in the following.

\begin{corollary}
\label{cor:hittingset-sampling-logn}
Let $\mathcal{S}^+ = \{\text{the first $10 \log N/p$ elements of $S_i$} \mid |S_i| \geq 10 \log N / p\}$ and $\mathcal{S}^{-} = \mathcal{S} \setminus \mathcal{S}^+$. Run the algorithm twice: once for $\mathcal{S}^-$ with the same set of weights as before and once on $\mathcal{S}^+$ by setting all weights to $N^2$. Let the output of these two runs be $H^-$ and $H^+$. Then:
\begin{equation*}
    \Phi^p_{\mathcal{S}}(H = H^- \cup H^+) \leq 4.
\end{equation*}
Each run takes at most $O(\log N)$ iterations. Moreover, all sets in $\mathcal{S}$ with size at least $10\log N/p$ are hit by $H$.
\end{corollary}

\subsection{Implementation}
\label{subsec:hittingset-impl}
The remaining piece of \Cref{thm:hittingset-sampling} is to find $H^t$ such that $f^t(H^t) \leq \E[f^t(G^t)]$. We first start with the construction of a suitable pairwise distribution. 

\paragraph{Construction of Pairwise Independent Distribution.} From the algorithm of the previous section, we need a pairwise distribution $\mathcal{P}$ on $n$ binary random variables $X_1, \dots, X_n \in \{0,1\}$ with bias $q$. Assume that $q = 2^{-\ell}$ for some $\ell \in \mathbb{N}$ and $n$ is a positive integer of the form $n = 2^{m}-1$ for $m \in \mathbb{N}$. We use the pairwise distribution that is used in~\cite{luby1993removing, berger1989efficient} which has a random seed of length $\ell m = O(\log 1/p \cdot \log n)$. Let us quickly recall the construction. We first assign an $\ell$-bit label $L_i$ to each $X_i$. Then, we set $X_i$ to one if and only if all the $\ell$ bits of $L_i$ is one. To construct the labels, we decompose the random seed $R$ into $\ell$ groups each containing $m$ bits as follows:
\begin{equation*}
    R = r^0_0\dots r^0_{m-1}r^1_0\dots r^1_{m-1}\dots r^{\ell-1}_{0}r^{\ell-1}_{m-1}
\end{equation*}
The $j$-th group $r^{j}_0\dots r^{j}_{m-1}$ is for constructing the $j$-bit of $L_i$s. To define $L_i(j)$ (the $j$-th bit of $L_i$), we use the bit representation of $i$. Suppose $i = \sum_{k=0}^{m-1} b_k 2^k$. Then:
\begin{equation*}
L_i(j) = b_0 r^{j}_0 \oplus \dots \oplus b_{m-1} r^{j}_{m-1}
\end{equation*}
This completes the construction. In the course of derandomization, we fix the random seed bit by bit. Suppose we fix the first $B$ bits of $R$ to $b_0,\dots,b_{B-1} \in \{0,1\}$. This gives us a new distribution $\mathcal{Q}$. The following result by Berger, Rompel, and Shor~\cite{berger1989efficient} is an important tool to achieve work-efficient derandomization.

\begin{lemma}[\cite{berger1989efficient}, Section 3.2]
\label{lem:berger-etal}
For any given subset $A \subseteq [n]$, we can compute
\begin{equation*}
    \sum_{i \in A} \E_{\mathcal{Q}}[X_i],\quad
    \sum_{i \in A} \sum_{j \in A} \E_{\mathcal{Q}}[X_i X_j]
\end{equation*}
with $O(|A|)$ processors and in $O(\log n)$ depth in the \pram model. In particular, we can compute these two quantities in $O(|A| \log n)$ time in the standard model.
\end{lemma}

\paragraph{Bit Fixing.}
Suppose we are in iteration $t$ and we want to find $H^t$ such that $f^t(H^t) \leq \E[f^t(G^t)]$. Suppose $\mathcal{P}^t$ is $\mathcal{P}$ as described above. If $q$ is not a power of two (which is needed for the pairwise construction), replace it with a power of two in the range $[q,2q)$. We can observe that for any $H$
\begin{equation}
\label{eq:enlarge-q}
    \Phi^{2p}_{\mathcal{S}}(H) \geq \Phi^{p}_{\mathcal{S}}(H)/2
\end{equation}
So with this replacement, we lose at most a two factor in the final bound for the potential function. Now, we start to fix the bits of the random seed of $\mathcal{P}^t$. Suppose we already fixed the first $B$ bits of the random seed $R$ by $b_0,\dots,b_{B-1}$. Let $e_x = \E[f(G^{t+1}) \mid R(0)=b_0,\dots,R(B-1)=b_{B-1},R(B)=x]$ for $x \in \{0,1\}$. If $e_0 \leq e_1$, then we fix $b_B$ to zero. Otherwise, we fix it to one. Suppose all the $\ell m$ bits are fixed and suppose that the random variable $X_i$ is $v_i \in \{0,1\}$ when we set the random seed to $b_0 \dots b_{\ell m - 1}$. Then, we set $H^t$ to $\{ i \in [n] \mid v_i = 1\}$. We can easily observe that $f^t(H^t) \leq \E[f^{t}(G^t)]$.

\paragraph{\pram Model.} We have all the ingredients for implementing the algorithm in the \pram model. This leads to the following theorem.

\begin{theorem}
\label{thm:hittingset-pram}
There is a deterministic algorithm that solves the hitting set problem by finding a subset $H$ with $\Phi^{p}_{\mathcal{S}}(H) \leq 4$ and with $\tO(\sum_{i=1}^{N} |S_i|)$ work and
\begin{equation*}
    O(\lceil p\Delta\rceil \cdot \log 1/p \cdot \log^2 n)
\end{equation*}
depth in the \pram model. Moreover, there is a deterministic algorithm that finds a subset $H$ with $\Phi^{p}_{\mathcal{S}}(H) \leq 8$ and such that $H$ hits all $S_i$s with size greater than $10 \log N / p$. This algorithm runs with $\tO(\sum_{i=1}^{N} |S_i|)$ work and
\begin{equation*}
    O(\log N \cdot \log 1/p \cdot \log^2 n)
\end{equation*}
depth in the \pram model.
\end{theorem}
\begin{proof}
The first algorithm is based on \Cref{thm:hittingset-sampling} and the second algorithm is based on \Cref{cor:hittingset-sampling-logn}. In those two algorithms, the potential function is upper bounded by $2$ and $4$. Here, we can only guarantee $4$ and $8$. This is because $q$, the sampling probability of one iteration, may not be a power of two. As discussed before (see \Cref{eq:enlarge-q}), we can handle this issue by paying an extra factor two in the approximation factor. In one iteration, we have $O(\log 1/p \cdot \log n)$ bit fixing. For each bit, we need to compute two conditional expectation which takes $O(\log n)$ depth and $O(\log n \cdot \sum_{i=1}^{N} |S_i|)$ work using  \Cref{lem:berger-etal}. Multiplying the number of iterations gives us the claimed bounds.
\end{proof}

\paragraph{\congest Model.}
First, let us describe how the hitting set problem is represented in the distributed model. Consider an $(N+n)$-node bipartite network $G = (A \sqcup B, E)$ where $A = [N]$ and $B = [n]$. A node $i \in A$ represents set $S_i$ and a node $j \in B$ represents element $j \in [n]$. There is an edge between $i \in A$ and $j \in B$ if and only of $j \in S_i$. We assume that $p$, $n$, and $\tau_{\mathcal{S}}^p$ (or an upper bound of it) is known to all nodes.

To simulate global decision making, we use $3$-separated network decomposition. We need to execute the following operation fast: For an arbitrary color $j$, let $C_1, \dots, C_d$ be the set of clusters with color $j$ in the given $3$-separated network decomposition. Suppose that each node $v$ in $C_1 \cup \dots \cup C_d$ knows a value $a_v$. For each cluster $C_i$, we want to broadcast the value $\sum_{v \in C_i} a_v$ to all nodes in $C_i$. We denote the round complexity of executing this operation for all clusters $C_1, \dots, C_d$ by $T^{\mathrm{agg}}_{\mathrm{ND}}$.

\begin{theorem}
\label{thm:hittingset-congest}
Given a $Q$-color $3$-separated network decomposition with aggregation time $T^{\mathrm{agg}}_{\mathrm{ND}}$ (as described above), there is a deterministic algorithm that solves the hitting set problem by finding a subset $H$ with $\Phi_{\mathcal{S}}^{p}(H) \leq 4$ in 
\begin{equation*}
     O(\lceil p\Delta \rceil \cdot Q \cdot \log 1/p \cdot \log n \cdot T_{\mathrm{ND}}^{\mathrm{agg}})
\end{equation*}
rounds of the \congest model. Moreover, there is a deterministic algorithm that finds a subset $H$ with $\Phi_{\mathcal{S}}^{p}(H) \leq 8$ and such that $H$ hits all $S_i$s with size greater than $10 \log N / p$. This algorithm runs in
\begin{equation*}
    O(\log N \cdot Q \cdot \log 1/p \cdot \log n \cdot T_{\mathrm{ND}}^{\mathrm{agg}})
\end{equation*}
rounds of the \congest model.
\end{theorem}

\begin{proof}
We want to derandomize iteration $t$. In contrast to the \pram model \Cref{thm:hittingset-pram}, in the \congest model, we do not have global communication and so we cannot decide which bit should be fixed in a straightforward way. However, we can simulate such global decision-making with network decomposition paying an extra factor $Q$ in the round complexity. For each cluster $C$, we independently draw a sample from the pairwise-independent distribution $\mathcal{P}$ with bias $q$. Recall that the input graph is a bipartite graph $G = (A \sqcup B, E)$. These samples assign a binary value to each node of $B$. Observe that the assigned values are also pairwise independent since the product of pairwise independent distributions is pairwise-independent. Now, to derandomize, we go through the colors one by one. Suppose we are working on color $j \in [Q]$ with $d$ clusters $C_1, \dots, C_d$. Moreover, suppose the first $b$ bits of random seeds of $C_1, \dots, C_d$ are fixed. We fix the $(b+1)$-th bit. Let us emphasize that each cluster has its own random seed and different clusters may fix the $(b+1)$-th bit differently. Consider cluster $C_i$ and a node $v \in A$ that is either in $C_i$ or is in the boundary of $C_i$ (i.e., $v$ is not in $C_i$ but has a neighbor in $C_i$). So this node represents a set $S_v$ in the corresponding hitting set problem. We assign two values $a^0_v$ and $a^1_v$ to $v$ where $a^x_v$ corresponds to the case when we fix the $(b+1)$-bit of the random seed of $C_i$ to $x$. Note that each neighbor of $v$ represents an element of $S_v$. If $v$ has a neighbor in the clusters with color $\{1,\dots,j-1\}$ that is already decided to be in $H$ (our final hitting set), then we set $a^x_v$ to zero. So suppose this is not the case and let $d$ be the number of neighbors of $v$ that are not in $C_i$ and are in a cluster with color in $\{j+1, \dots, Q\}$. Then, we set $a_v^x$ to $$\frac{(1 + \binom{d}{2}q^2 - dq) \cdot F_v^x \cdot w_v e^{-|S_i|(T - t) p /T}}{\tau_{\mathcal{S}}^{p}}$$ where $F_v^b$ is 
\begin{equation*}
    F_v^x = \E[1 + \sum_{u \in C_i \cap B: u\in S_v} \sum_{u \in C_i \cap B: u\in S_v \wedge u < u'} X_u X_{u'} - \sum_{u \in C_i \cap B: u\in S_v} X_u \mid \text{first $b$ bits and $(b+1)$-th bit is $x$}]
\end{equation*}
where $X_u$ represents the indicator random variable of element $u$. Note that the given network decomposition is $3$-separated and so all the boundaries of $C_1, \dots, C_d$ are disjoint. So $v$ can compute $F_v^b$ in $\tO(|S_v|)$ according to \Cref{lem:berger-etal}. Also, note that that this gives us the contribution of $S_v$ to
\begin{equation*}
    \E[f^{t}(\cdot) \mid \text{first $b$ bits and $(b+1)$-th bit is $x$}]
\end{equation*}
Next, for each element $u \in A_i$, set $a_u^x$ to 
\begin{equation*}
    \frac{\E[X_u \mid \text{first $b$ bits and $(b+1)$-th bit is $x$}]}{4np}.
\end{equation*}
In the end, for each cluster $C_i$, we compute two values $e_i^x$ for $x \in \{0,1\}$ which is $$\sum_{v \in C_i \cup (\partial(C_i) \cap B)} a_v^x$$ where $\partial(C_i)$ denotes the boundary of $C_i$. We broadcast $e^b_i$ to each nodes in $C_i$. This can be done in $O(T_{\mathrm{ND}}^{\mathrm{agg}})$ rounds for all $C_i$s simultaneously. Next, nodes of $C_i$ set the $(b+1)$-bit of the random seed to zero if $e^0_i \leq e^0_i$ and set it to one otherwise. This completes the bit fixing.

There are $T$ sampling iterations (if we apply \Cref{thm:hittingset-sampling}, $T = \lceil p\Delta\rceil$, and if we apply \Cref{cor:hittingset-sampling-logn}, $T = O(\log N)$), $Q$ colors, and $O(\log 1/p \cdot \log n)$ bits to fix for each color. Multiplying these numbers gives us the number of bit fixing. Taking into account that fixing each bit takes $O(T_{\mathrm{ND}}^{\mathrm{agg}})$ rounds of the \congest model concludes the proof.
\end{proof}

\begin{corollary}
\label{cor:hittingset-congest-logn}
There is a deterministic algorithm that solves the hitting set problem by finding a subset $H\subseteq [n]$ with $\Phi^p_{\mathcal{S}}(H) = O(1)$ in $\poly(\log n)$ rounds of the \congest model and with total computations $\tO(m)$.
\end{corollary}

\begin{proof}
There is a work-efficient deterministic algorithm for finding a $3$-separated $O(\log n)$-color network decomposition in $\polylog(n)$ rounds and with $T_{\mathrm{ND}}^{\mathrm{agg}} = \polylog(n)$(see Theorem 2.12 of Rozho\v{n} and Ghaffari~\cite{rozhonghaffari20}). Plugging this bound in \Cref{thm:hittingset-congest} concludes the proof.
\end{proof}

\section{Applications of Hitting Set}
In this section, we discuss two applications of the hitting set problem. One is the distributed construction of multiplicative spanners and the other is the parallel construction of distance oracles. Let us quickly define these notions. A subgraph $H = (V,E') \subseteq G = (V,E)$ is an $\alpha$-spanner of $G$ if for all pairs of nodes $u, v \in V$, we have:
\begin{equation*}
    d_G(u,v) \leq d_H(u,v) \leq \alpha \cdot d_G(u,v). 
\end{equation*}
A distance oracle is a data structure that accepts a pair of nodes $(u,v)$ as a query and returns their distance in $G$. In \Cref{subsec:distance-oracles}, we discuss \textit{source-restricted approximate distance oracle} in which $s$ nodes of $G$ are marked as source and it is guaranteed that $u$ is always a source. The term ``approximate'' allows the oracle to return an approximation of $d_G(u,v)$ rather than its exact value.
\label{sec:hittingset-app}
\subsection{Spanners}
\label{subsec:spanners}

\begin{theorem}
There is deterministic algorithm in $\poly(\log n)$ rounds of the \congest model and with total computations $\tO(m)$ that finds a $(2k-1)$-spanner with $O(nk + n^{1 + 1/k} \log k)$ and $O(nk + n^{1 + 1/k} k)$ edges for unweighted and weighted graphs, respectively.
\end{theorem}

\begin{proof}
We derandomize Baswana-Sen algorithm~\cite{baswana2007simple}. Let us quickly recall this algorithm. It consists of $k$ steps. The input of step $i$ is a clustering denoted by $\mathcal{C}_i$. Each cluster has a center node known to all of its members. The input of the first step is the trivial clustering: there are $n$ clusters each containing a single node. During one step, we sample some of the clusters, and then based on that sampling, some nodes stay in their clusters, some get unclustered, and some join other clusters. After this, the current step $i$ terminates, and the new clustering $\mathcal{C}_{i+1}$ is passed to the next step. Here is what we do in step $i$ for $i \leq k-1$ (we discuss the last step, $i$ equals $k$, later):

\begin{enumerate}
    \item Each cluster of $\mathcal{C}_i$ is sampled with probability $p = n^{-1/k}$.
    \item A node that is in a sampled cluster, stays put in its own cluster.
    \item For a node $v$ in an unsampled cluster, let $C_1, \dots, C_d$ be the set of clusters containing at least one neighbor of $v$. Let $e_i = \{u_i \in C_i, v\}$ be an edge with the minimum weight between $v$ and one of the nodes in $C_i$. If there are several edges with the minimum weight, $v$ selects one of them arbitrarily. Let $w_i$ be the weight of $e_i$. Without loss of generality, suppose $w_1 \leq \dots \leq w_d$.  If all of $C_1, \dots, C_d$ are unsampled, $v$ adds all edges $e_1, \dots, e_d$ to the output spanner and gets unclustered. Otherwise, let $j$ be the minimum index for which $C_j$ is sampled. Then, $v$ adds $e_1, \dots, e_{j}$ to the output spanner and joins the sampled cluster $C_j$. Note that all such $v$ runs this step simultaneously.
\end{enumerate}
In the last step, we do the exact same thing except that we sample no cluster (each cluster is sampled with probability zero rather than $n^{-1/k}$).

The output of Baswana-Sen is always a $(2k-1)$-spanner and only the size of the output depends on the randomness. From the algorithm description, you can see that the only randomized part of the Baswana-Sen algorithm is the sampling of clusters. Our goal is to find the set of sampled clusters of each step deterministically. If we have the following properties on the set of sampled clusters, then we can guarantee the claimed bounds on the size of the output spanner (see~\cite{bezdrighin2022deterministic}, Lemma 3.3):
\begin{enumerate}[(a)]
    \item For each $i$, the number of clusters in $\mathcal{C}_i$ is at most $n^{1-(i-1)/k}$.
    \item The number of edges added to the output spanner is bounded as follows: For the unweighted case, the total number of edges added by nodes with at least $\gamma_1 n^{1/k} \log k$ neighboring clusters for a large enough constant $\gamma_1 > 0$ is at most $O(n^{1 + 1/k}/k)$. For the weighted case, all nodes add at most $O(n^{1 + 1/k})$ edges to the output.
    \item A node that is clustered in $\mathcal{C}_i$, remains clustered if it has at least $\gamma_2 n^{1/k} \log n$ neighbouring clusters for a large enough constant $\gamma_2 > 0$.
\end{enumerate}
We can frame these properties as a hitting set problem. To avoid cluttering the notation, we refer to the universe size in the corresponding hitting set problem of step $i$ by $n^h_i$ and its number of sets by $N^h_i$. In step $i$, we have the following hitting set problem: There is an element in the universe for each cluster in $\mathcal{C}_i$. So $n^h_i = |\mathcal{C}_i| \leq n$. For each clustered node $v$ in $\mathcal{C}_i$, there is a set $S_v$ containing all of its neighboring clusters. So $N^h_i \leq n$. The parameter $p$ for the hitting set problem is set to the sampling probability of Baswana-Sen divided by a large enough constant $\gamma_3 > 0$, i.e., $p = n^{-1/k}/\gamma_3$ (note that the last step is already deterministic and no derandomization is needed there). For unweighted graphs, we set the weight of $S_v$ to its size $w_v = |S_v|$. For weighted graphs, we consider the hitting ordered set problem as discussed in \Cref{lem:hittingset-ordered-reduction}. For each clustered node $v$, we assign the order $\pi_v(\cdot)$ on $S_v$. Suppose that the neighboring clusters of $v$ are $C_1, \dots, C_d$ and the minimum weight of an edge between $C_i$ and $v$ is $w_i$. Then $C_i$ comes before $C_j$ in $\pi_v(\cdot)$ if $w_i < w_j$ or $w_i = w_j$ and $i < j$.

With straightforward calculations, we can see that all the three required properties are satisfied if we solve the presented hitting set problem with \Cref{cor:hittingset-congest-logn} (for the hitting ordered set problem, we first use the reduction \Cref{lem:hittingset-ordered-reduction}).

We have $k \leq \log n$ steps in total. As described above, each step can be derandomized by solving a hitting set problem. So the total round complexity is $\poly(\log n)$ by applying \Cref{cor:hittingset-congest-logn}. One issue here is that each element in the defined hitting set problem corresponds to a cluster. This issue can be handled by contracting each cluster to a node and using the fact that the network decomposition of~\cite{rozhonghaffari20} also works on contracted graphs. This slows down the round complexity only by a factor $k = O(\log n)$ as each cluster has diameter $k$.
\end{proof}

\begin{theorem}
For any $\eps > 0$, there is deterministic distributed algorithm in $\poly(\log n) / \eps$ rounds of the \congest model and with total computations $\tO(m)$ that finds a spanner with size $n(1 + \varepsilon)$ and with stretch $O(\log n \cdot 2^{\log^* n} / \eps)$ and $O(\log n \cdot 4^{\log^* n} / \eps)$ stretch for unweighted and weighted graphs, respectively.
\end{theorem}

\begin{proof}
We derandomize the algorithm of Pettie~\cite{pettie2010distributed} to get a spanner with $O(n)$ edges and with stretch $O(\log n \cdot 2^{\log^* n})$ and $O(\log n \cdot 4^{\log^*} n)$ for unweighted and weighted graphs, respectively. Pettie's algorithm is combining $O(\log^* n)$ application of Baswana-Sen back to back and the hitting set problem we encounter in Pettie's algorithm, is exactly the same as the Baswna-Sen. So we do not repeat this here. We refer interested readers to Theorem 1.5 of ~\cite{bezdrighin2022deterministic} where the full algorithm and a slower derandomized version of it is discussed. Let us note that the original algorithm of Pettie only works for unweighted graphs, but with a simple modification which is proposed in~\cite{bezdrighin2022deterministic}, it can work on weighted graphs as well. To reduce the number of edges from $O(n)$ to $n(1 + \eps)$, we apply the deterministic reduction of ~\cite{bezdrighin2022deterministic}, Theorem 1.2. 
\end{proof}

\subsection{Approximate Distance Oracles}
\label{subsec:distance-oracles}
This section is devoted to the parallel implementation of the approximate distance oracle by Roditty, Thorup, and Zwick~\cite{roditty2005deterministic}. There, given a weighted graph $G = (V,E)$, a stretch parameter $k$, and a set of $s$ sources $S \subseteq V$, they deterministically construct a data structure of size $O(kns^{1/k})$ and in $\tO(ms^{1/k})$ time. For a query $(u,v)$, the data structure can compute a value $q$ such that
\begin{equation*}
    d(u,v) \leq q \leq (2k-1) d(u,v)
\end{equation*}
in $O(k)$ time. See \Cref{alg:distance-oracle} for their algorithm for constructing the data structure and \Cref{alg:distance-oracle-query} for how they evaluate a query.

\begin{algorithm}
	\caption{Approximate Distance Oracle~\cite{roditty2005deterministic}}
	\label{alg:distance-oracle}
	\begin{algorithmic}[1] 
		\Procedure{DistOracle}{G, k} 
		\State $A_0 = S, A_k = \emptyset$.
		\State $\ell = 10s^{1/k} \log n$.
		\For{$i=1,\dots,k-1$}
		    \State For each $v \in V$, find $p_i(v) \in A_{i-1}$ such that $d(p_i(v), v) = d(A_{i-1}, v).$
			\State For every $v \in V$, compute $N_{i-1}(v)$ which is  the set of $\ell$ closest nodes to $v$ in $A_{i-1}$.
			\State Find a set $A_i \subseteq A_{i-1}$ such that:
			    \Statex $\quad \quad \quad \quad$ (a) $|A_i| \leq s^{1-i/k}$.
			    \Statex $\quad \quad \quad \quad$ (b) $A_i$ hits $N_{i-1}(v)$ for all $v \in V.$
			    \Statex $\quad \quad \quad \quad$ (c) $\sum_{v \in V} |\{w \in A_{i-1} - A_{i} \mid d(w,v) < d(A_i, v)\}| = O(ns^{1/k}).$
		\EndFor
		\State For each $v \in V$, compute $p_{k-1}(v)$.
		\State For every $v \in V$, set $B(v) = A_{k-1}$. 
		\For{$i=0,\dots,k-2$}
		    \State For every $v \in V$, set $B(v) = B(v) \cup \{w \in N_i(v) \mid d(w,v) < d(A_{i+1}, v)\}$.
		\EndFor
		\State For each $v \in V$, create a hash table $H(v)$ with an entry $(v,d(v,w))$ for each $w \in B(v)$.
		\EndProcedure
	\end{algorithmic}
\end{algorithm}

\begin{algorithm}
	\caption{Evaluating a query~\cite{roditty2005deterministic}}
	\label{alg:distance-oracle-query}
	\begin{algorithmic}[1] 
		\Procedure{Query}{$u \in S$,$v$} 
		\State $w = u$, $i = 0$.
		\While{$w \not \in B(v)$}
		    \State $i = i + 1$.
		    \State $(u,v) \leftarrow (v,u)$.
		    \State $w \leftarrow p_i(u)$
		\EndWhile
        \Return $d(w,u) + d(w,v)$
		\EndProcedure
	\end{algorithmic}
\end{algorithm}

\begin{theorem}
\label{thm:distance-oracle}
Given an undirected weighted graph $G = (V,E)$, a set of $s$ sources $S \subseteq V$, stretch parameter $k$, and error $\eps > 0$, there is a deterministic algorithm that solves the source-restricted distance oracle problem with $\tO_{\eps}(ms^{1/k})$ work and $\tO_{\eps}(\poly(\log n))$ depth in the \pram model. The data structure has size $O(nks^{1/k})$ and for each query $(u,v)$, the oracle can return a value $q$ in $O(k)$ time that satisfies
\begin{equation*}
    d(u,v) \leq q \leq (2k-1)(1 + \eps)d(u,v).
\end{equation*}
\end{theorem}

\begin{proof}
It is enough to provide a parallel algorithm with $\tO_{\eps}(\poly(\log n))$ depth for computing $A_i$, $N_i(\cdot)$, and the hash table. This gives us all the ingredients we need to run the algorithm.

Note that finding a suitable $A_i$ in \Cref{alg:distance-oracle} is just an instance of hitting ordered set problem and we can apply \Cref{lem:hittingset-ordered-reduction} and \Cref{cor:hittingset-sampling-logn}. The universe is $A_{i-1}$ and for each $v \in V$, we want to hit the set $N_{i-1}(v)$. We also need to determine $\pi_{i,v}(\cdot)$. An element $w$ comes before $w'$ in this order if $d(w,v) < d(w',v)$. If the distances are equal, we break the tie based on the identifier of $w$ and $w'$. If we set the sampling probability to $p = s^{-1/k}/\gamma$ for a large enough constant $\gamma > 0$ (indeed $\gamma = 24$ is enough), then we can compute a suitable $A_i$ satisfying all the three required properties with $\tO(m)$ work and $O_{\eps}(\poly(\log n))$ depth in the \pram model using \Cref{thm:hittingset-pram} and the reduction \Cref{lem:hittingset-ordered-reduction}.

In~\cite{roditty2005deterministic}, they compute $N_i(\cdot)$ by running $\ell$ instances of Single Source Shortest Path problem (SSSP). There is no known parallel algorithm for SSSP with poly-logarithmic depth. However, recently, Rozhoň et al.~\cite{rozhovn2022undirected} proposed a work-efficient algorithm for computing $(1+\eps)$-approximation of SSSP with poly-logarithmic depth. We can replace the exact computation with an approximation, losing $(1+\varepsilon)$ in the final stretch guarantee.

For computing the hash tables, we can apply the construction of Alon and Naor~\cite{alon1996derandomization}. There, they provide a deterministic hash table of $t$ elements into $O(t)$ space with read access of $O(1)$ time. While they did not discuss the parallel implementation of their construction, their algorithm can be implemented in $\poly(\log n)$ depth in a straightforward way. Their approach is derandomizing a randomized hash function using the method of conditional expectation on epsilon-biased spaces. They define a potential function (see section 3.1. of~\cite{alon1996derandomization}) which is a simple aggregation and can be parallelized. We do not discuss the full details as the implementation is straightforward.  
\end{proof}

\newpage




 \section*{Acknowledgments}
M.G., C.G., S.I., and V.R. were supported in part by the European Research Council (ERC) under the European Unions Horizon 2020 research and innovation program (grant agreement No.~853109) and the Swiss National Science Foundation (project grant 200021\_184735). B.H. was supported in part by NSF grants CCF-1814603, CCF-1910588, NSF CAREER award CCF-1750808, a Sloan Research Fellowship, funding from the European Research Council (ERC) under the European Union's Horizon 2020 research and innovation program (grant agreement 949272), and the Swiss National Science Foundation (project grant 200021\_184735).


\bibliographystyle{alpha}
\bibliography{refs}

\end{document}